\documentclass[11pt,letterpaper]{llncs}

\usepackage{fullpage}
\usepackage[numbers,sort]{natbib}

\usepackage{comment}
\usepackage{graphicx}
\usepackage{tikz}
\usetikzlibrary{arrows,automata}

\usepackage{amssymb}

\newcommand{\red}[1]{\textbf{\textcolor{red}{#1}}}
\newcommand{\blue}[1]{\textbf{\textcolor{blue}{#1}}}

\newcommand{\preord}{\hat}
\newcommand{\wheel}{\bar}
\newcommand{\nblocks}{r'}

\begin{document}
\title{On Locating Paths in Compressed Tries}
%
%\titlerunning{Abbreviated paper title}
% If the paper title is too long for the running head, you can set
% an abbreviated paper title here
%
\author{Nicola Prezza
%\inst{1}
\orcidID{0000-0003-3553-4953}} 
%\and Second Author\inst{2,3}\orcidID{1111-2222-3333-4444} \and Third Author\inst{3}\orcidID{2222--3333-4444-5555}}
%
\authorrunning{N. Prezza}
% First names are abbreviated in the running head.
% If there are more than two authors, 'et al.' is used.
%
\institute{Ca' Foscari University of Venice, Italy
\email{nicola.prezza@unive.it}}
\maketitle              % typeset the header of the contribution
\begin{abstract}
In this paper, we consider the problem of compressing a trie while supporting the powerful \emph{locate} queries: to return the pre-order identifiers of all nodes reached by a path labeled with a given query pattern. Our result builds on top of the XBWT tree transform of Ferragina et al. [FOCS 2005] and generalizes the \emph{r-index} locate machinery of Gagie et al. [SODA 2018, JACM 2020] based on the run-length encoded Burrows-Wheeler transform (BWT). Our first contribution is to propose a suitable generalization of the run-length BWT to tries. We show that this natural generalization enjoys several of the useful properties of its counterpart on strings: in particular, the transform natively supports counting occurrences of a query pattern on the trie's paths and its size $r$ captures the trie's repetitiveness and lower-bounds a natural notion of trie entropy. 
Our main contribution is a much deeper insight into the combinatorial structure of this object. In detail, we show that a data structure of $O(r\log n) + 2n + o(n)$ bits, where $n$ is the number of nodes, allows locating the $occ$ occurrences of a pattern of length $m$ in nearly-optimal $O(m\log\sigma + occ)$ time, where $\sigma$ is the alphabet's size. Our solution consists in sampling $O(r)$ nodes that can be used as "anchor points" during the locate process. Once obtained the pre-order identifier of the first pattern occurrence (in co-lexicographic order), we show that a constant number of constant-time jumps between those anchor points lead to the identifier of the next pattern occurrence, thus enabling locating in optimal $O(1)$ time per occurrence. 
\end{abstract}

\newpage

\clearpage
\pagenumbering{arabic} 

\section{Introduction}
	
A compressed text index is a data structure representing a text $T$ within compressed space and supporting fast \emph{count} and \emph{locate} queries: given a query pattern, count/return all positions in $T$ where the pattern occurs~\cite{survey}. 
The first compressed indexes date back twenty years and operate within a space bounded by the text's empirical entropy~\cite{FMI,CSA}. 
Entropy, however, does not capture long repetitions: entropy-compressing $T\cdot T$ yields an archive twice as big as the entropy-compressed $T$~\cite{KN13}. For this reason, in recent years more powerful compressed indexes have emerged; these are based on the Lempel-Ziv factorization~\cite{KN13}, the run-length Burrows-Wheeler Transform (BWT)\cite{r-index,RLFM,MNSV08}, context-free grammars~\cite{CNspire12} and, more recently, string attractors~\cite{navarro2019universal,attractors}. 
%Compressed text indexes have had a dramatic impact in domains such as bioinformatics~\cite{langmead2009ultrafast,li2009fast}. The recent rise of massive repetitive datasets, however, advocates for compressed indexes able to handle even more structured data such as labeled graphs~\cite{pangenomics}.
In this domain, the state of the art compressed index is represented by the so-called \emph{r-index} of Gagie et al. \cite{r-index}. This index takes a space proportional to the number $r$ of equal-letter runs in the BWT and  locates pattern occurrences in log-logarithmic time each, being orders of magnitude faster than all comparably-small alternatives  in practice.
On trees, the state of the art is much less mature.  While some of the above compression techniques have been extended to trees~\cite{gawrychowski2016lz77,busatto2008efficient}, less is known about tree indexing within compressed space. 
%Jacobson~\cite{jacobson1989space} showed how to represent the topology of a tree with $n$ nodes in worst-case optimal $2n+o(n)$ bits of space while also supporting basic navigation queries in constant time. Subsequent works~\cite{suctrees,RRR} added many navigation functionalities to this (as well as other) succinct tree representation.
Ferragina et al.~\cite{XBWT} have been the first to 
tackle the tree indexing problem: their \emph{XBW Transform} (XBWT in the following) stores any labeled tree within entropy-compressed space while also supporting fast \emph{count} queries on it. Crucially, they did not discuss how to \emph{locate} paths labeled with a given pattern. In this setting, a natural generalization of the problem is to return the pre-order identifier of all nodes reached by the query pattern. \vspace{3pt} \newline
\textbf{Our Contributions}
In this paper, we show for the first time 
how to support the powerful locate queries on \emph{compressed} tries.
To begin with, we generalize the notion of run-length encoding to the XBWT of a trie and show that the number $r$ of runs in the XBWT is a valid compressibility measure as it captures the trie's repetitiveness and it lower bounds the $k$-th order worst-case entropy $\mathcal H^{wc}_k$ of the trie. 
%As observed above, however, entropy does not capture large repetitions. We therefore extend the concept of string attractors~\cite{attractors} to labeled trees and show that the XBWT induces a tree attractor of size $r$. We furthermore relate $r$ with the size $\omega$ of the smallest equivalent Wheeler automaton~\cite{SODA20,GAGIE201767} by showing $r \leq \omega$. 
%An intriguing consequence of this inequality is that one can reduce the problem of indexing any acyclic Wheeler automaton $\mathcal A$ to the problem of indexing (the run-length XBWT of) the equivalent tree within $O(r)$ words of space: the resulting index will not be larger than $\mathcal A$.
%To conclude, we generalize the powerful \emph{locate} queries to trees: to return the pre-order identifier of all nodes reached by a path labeled with a given pattern $P$.
Our main contribution is a deep insight into the combinatorial structure of the run-length XBWT and leads to a neat (and nontrivial) generalization of the r-index to tries. 
We first observe that the standard sampling mechanism of compressed suffix arrays can easily be extended to the XBWT. This simple solution, however, requires also a sampling of $O((n/t)\log n) + o(n)$ bits \emph{on top of the XBWT} to support $\tilde O(t)$-time locate queries.
The problem with this sampling is that it does not depend on the structural properties of the underlying trie. We show that it is indeed possible to design a more advanced sampling mechanism that depends on the combinatorial properties of the XBWT.
%: by sampling $O(r)$ particular nodes, we give a more elegant and efficient solution to the locate problem. 
%Note that our sampling takes compressed space. 
In detail, our machinery uses a repetition-aware sampling of size $O(r)$ and locates nodes in two steps: (1) during the counting process, we locate the pre-order identifier of the co-lexicographically smallest node $u$ reached by the query pattern, and (2) we show that a constant number of constant-time "jumps" between the sampled nodes is sufficient to locate the co-lexicographic successor of $u$. By repeating this process $occ$ times ($occ$ being the number of pattern's occurrences), we manage to locate all pattern occurrences in constant time each. 
Our data structure takes $O(r\log n) + 2n + o(n)$ bits of space, where the linear overhead is required to support constant-time queries on the trie's topology. 
While we focus on tries only, our results can be generalized to arbitrary labeled trees; since the primary goal of this paper is to provide a useful combinatorial insight into the run-length XBWT, we preferred to stick to the trie case which is simpler to introduce. A natural improvement over our work would be to compress the tree topology within $O(r\log n)$ bits of space as well. We believe that this should be possible by unveiling further combinatorial properties of the run-length XBWT.

%The first ten pages of this manuscript contain a concise description of all our contributions and can be followed by non-specialists. 
%The proofs of all claims can be found in the appendix.

\section{Definitions}

We work with edge-labeled tries $\mathcal T = (V,E)$ with $n$ nodes and labels from alphabet $\Sigma = \{1,\dots,\sigma\}$ ordered by a total order $\prec$. 
%In particular, in this setting the outgoing labels of a node form always a subset of $\Sigma$.
We extend  $\prec$ to $\Sigma^*$ using the co-lexicographic order (i.e. the strings' characters are compared right-to-left).
Given a string $S$, the number $rle(S)$ of equal-letter runs of $S$ is the number of maximal unary substrings of $S$ (for example, $rle(aaabbccaaa) = 4$).
We identify tree nodes by their pre-order identifier $\preord u$; node 1 is the root.
%by sorting its elements \emph{co-lexicographically}: letting $x,y\in\Sigma^*$ and $a,b\in \Sigma$, we have that $x\cdot a \prec y\cdot b$ if and only if ($a\prec b$) or ($a=b$ and $x\prec y$), where the empty string $\epsilon$ is co-lexicographically smaller than any element in $\Sigma^+$. 
%We will identify the tree's nodes with their pre-order number and use the notation $\preord u$ to identify such node representation. Later in the paper, we will introduce another node representation that will use a different notation to avoid confusion. 
Function $\pi(\preord u)$ returns the parent of node $\preord u$, and $\lambda(\preord u)$ indicates the label of the edge $(\pi(\preord u),\preord u)$. For the root, we take $\lambda(1) = \#$, where $\#$ is the lexicographically-smallest character in $\Sigma$, not labeling any edge.
Notation $\lambda(\Pi)$ denotes the string $\lambda(\preord u)\cdots \lambda(\preord v)$ labeling path  $\Pi = \preord u \rightsquigarrow \preord v$.
We assume the alphabet to be \emph{effective}: for each $c\in\Sigma$, there exists $\preord u$ such that $\lambda(\preord u) = c$.
Function $child_c(\preord u)$ returns the child of $\preord u$ reached by following the edge labeled $c$. 
If $\preord u$ does not have such a child, then $child_c(\preord u) = \bot$. 
We consider the children of each node to be implicitly sorted according to their incoming labels.
Function $out(\preord u)$ returns the (possibly empty) set $\{ c\ :\ child_c(\preord u)\neq \bot \}$ of the characters labeling the outgoing edges of $\preord u$.
Let $U\subseteq V$. The forest $\mathcal T(U)$ is the set of the subtrees of $\mathcal T$ induced by $U$. We say that $\mathcal T(U)$ is a \emph{subtree} if it is connected. A subtree $\mathcal T(U)$ with root $\preord u$ is \emph{complete} if $U$ contains all descendants of $\preord u$ in $\mathcal T$. 
The equivalence relation $\approx$ denotes isomorphism between (the complete subtrees rooted in) two nodes: $\preord u \approx \preord v$ if and only if, for each $c\in\Sigma$, $child_c(\preord u) \approx child_c(\preord v)$, where $\preord u \approx \bot$ if and only if $\preord u = \bot$.
In some of our results we will treat trees as deterministic finite state automata (DFA), with the root being the initial state and all states being final.
%For us, a DFA will be a quadruple $\mathcal A = (V,E,\preord s,\Sigma)$, where $V$ is a set of states (or vertices), $\Sigma$ is the alphabet (or set of labels), $E\subseteq V\times V \times \Sigma$ is a set of directed labeled edges, $\preord s \in V$ is the start state (or source), and all states are final (accepting). 
%We moreover require that $\preord s\in V$ is the only node with in-degree zero and that every node in $V$ is reachable from $\preord s$. 
%If $\mathcal A$ is a tree, then $\preord s$ is the root. 
%The \emph{language accepted by $\mathcal A$} is the set of strings that can be read from $\preord s$ to any node in $V$. 
%The \emph{minimum automaton recognizing the same language as $\mathcal A$} is the finite state automaton with the minimum number of states that recognizes the same language accepted by $\mathcal A$.
%Since all states are final, note that our isomorphism relation $\approx$ between nodes coincides with the Myhill-Nerode equivalence relation. 
%The \emph{quotient automaton} with respect to an equivalence relation $\equiv$ on $V$ is defined as $\mathcal A/_\equiv = (V/_\equiv, E/_\equiv,\Sigma)$, where 
%$E/_\equiv = \{([\preord u]_\equiv, [\preord v]_\equiv, c) :(\exists \preord u'\in [\preord u]_\equiv, \preord v'\in [\preord v]_\equiv)(  (\preord u',\preord v',c)\in E)\}$.
We work in the word RAM model with words of size $w = \Theta(\log n)$ bits. The space of our data structures will be given either in words or bits; in all cases we will clearly specify which unit of measurement we use.

\section{The Run-Length Encoded XBWT}\label{sec:XBWT}

We start our discussion with the problem of compressing tries.
Our solution is obtained by extending run-length encoding to the XBWT of Ferragina et al. \cite{XBWT-FOCS,XBWT}.
While the results presented in this section are generalizations of known constructions from strings to tries, they give us the basis for introducing our main contribution in the next section: a run-length compressed index for tries.

The XBWT is based on the idea of sorting \emph{co-lexicographically} the $n$ tree's nodes: we declare $\preord u < \preord v$ if and only if $\lambda(1 \rightsquigarrow \preord u) \prec \lambda(1 \rightsquigarrow \preord v)$.
Equivalently, the order $<$ satisfies the following two \emph{co-lexicographic axioms}: (i) if $\lambda(\preord u) \prec \lambda(\preord v)$ then $\preord u<\preord v$, and (ii)  if $\lambda(\preord u) = \lambda(\preord v)$ and $\pi(\preord u) < \pi(\preord v)$, then $\preord u<\preord v$. 
We will call $<$ the \emph{co-lexicographic} order of the tree's nodes.
%\emph{Wheeler graphs}~\cite{GAGIE201767} generalize this idea to all labeled graphs admitting such an ordering. 
%The set of Wheeler graphs is a strict superset of the set of all labeled trees. 
Let $\preord u_1 < \dots < \preord u_n$ be the sorted sequence of nodes. 
With $<_{pred}$ we denote the predecessor relation with respect to $<$: $\preord u_i <_{pred} \preord u_j$ if and only if $j = i+1$.
The subscripts in nodes $\preord u_1 < \dots < \preord u_n$ are the second node representation we will use in the paper: the co-lexicographic (co-lex for brevity) representation $\wheel u$ of (pre-order) node $\preord u_i$ is precisely $\wheel u = i$. 

%Intuitively, we will need two different node representations (pre-order, $\preord u$, and Wheeler-order, $\wheel u$) because our index will require particular tree navigation queries on both of them.

%In their original paper~\cite{XBWT}, Ferragina et al. consider general node-labeled trees (instead of cardinal edge-labeled), and define the XBWT as the string obtained by concatenating the nodes' labels after sorting the nodes in Wheeler order. Since this string alone is not sufficient to reconstruct the tree's topology, they also mark (in a bitvector) all nodes being the last among their siblings. In the following we give a completely equivalent definition on tries which avoids explicitly marking nodes and will be more convenient for our purposes. The \emph{XBWT} of a trie is defined as the sequence of all children's labels, in Wheeler order:

We now give a definition of the XBWT that (on tries) is completely equivalent to the original one given by Ferragina et al.~\cite{XBWT}. See Figures \ref{fig:example1} and \ref{fig:example} for a running example.

\begin{definition}[\cite{XBWT}]\label{def:XBWT}
 $XBWT(\mathcal T) = out(\preord u_1), out(\preord u_2), \dots, out(\preord u_n)$.
\end{definition}

For brevity, we shall simply write $XBWT$ instead of $XBWT(\mathcal T)$. The original trie $\mathcal T$ can be reconstructed from $XBWT$~\cite{XBWT}. 
%Figure \ref{fig:example1} depicts the trie used as running example throughout the paper. Figure \ref{fig:example} shows the XBWT of the trie (in addition to other components discussed later).

\begin{figure}
		\centering
		\tikzset{every state/.style={minimum size=16pt}}
		\begin{tikzpicture}[->,>=stealth', semithick, auto, scale=.7]
		\scriptsize
		\node[state] (root)    at (0,0)		{$\blue 1$};
		\node[state] (a)    at (-5,-1.5)		{$2$};

		\node[state] (aa)    at (-9,-3)		{$\red 3$};
		\node[state] (ab)    at (-5.5,-3)		{$\red{14}$};
		\node[state] (ac)    at (-3,-3)		{${18}$};

		\node[state] (acb)    at (-3.5,-4.5) {${19}$};
		\node[state] (acc)    at (-2.4,-4.5) {${20}$};

		\node[state] (acca)    at (-2.4,-6) {$\red{21}$};

		\node[state] (aba)    at (-6.1,-4.5)		{$\red{15}$};
		\node[state] (abc)    at (-4.9,-4.5)		{$\red{17}$};

		\node[state] (abab)    at (-6.1,-6)		{${16}$};

		\node[state] (aaa)    at (-10,-4.5)		{$4$};
		\node[state] (aab)    at (-9,-4.5)		{$6$};
		\node[state] (aac)    at (-8,-4.5)		{$\red7$};

		\node[state] (aaab)    at (-10,-6)		{$ 5$};

		\node[state] (aacb)    at (-8.5,-6)		{$\blue 8$};
		\node[state] (aacc)    at (-7.5,-6)		{$9$};

		\node[state] (aacca)    at (-7.5,-7.5)		{$\blue{10}$};

		\node[state] (aaccaa)    at (-8,-9)		{${11}$};
		\node[state] (aaccac)    at (-7,-9)		{${13}$};

		\node[state] (aaccaab)    at (-8,-10.5)		{${12}$};

		\node[state] (ba)    at (-1,-3)		{${23}$};
		\node[state] (bc)    at (1,-3)		{${25}$};

		\node[state] (bab)    at (-1,-4.5)	{${24}$};

		\node[state] (b)    at (0,-1.5)		{${22}$};
		\node[state] (c)    at (4,-1.5)		{$\red{26}$};
        \draw (root) edge [bend right=15, above] node {a} (a);
        \draw (root) edge [] node {b} (b);
        \draw (root) edge [bend left=15, above] node {c} (c);

        \draw (a) edge [bend right=10, above] node {a} (aa);
        \draw (a) edge [] node {b} (ab);
        \draw (a) edge [bend left=10, above] node {c} (ac);

        \draw (b) edge [bend right=10, above] node {\hspace{-3pt}a} (ba);
        \draw (b) edge [bend left=10, above] node {\hspace{3pt}c} (bc);

        \draw (ab) edge [bend right=10, above,red] node {\hspace{-4pt}\red a} (aba);
        \draw (ab) edge [bend left=10, above,red] node {\hspace{4pt}\red c} (abc);
        
        \draw (ba) edge [] node {b} (bab);

        \draw (aba) edge [red] node {\red b} (abab);

        \draw (aa) edge [bend right=10, above,red] node {\hspace{-3pt}\red a} (aaa);
        \draw (aa) edge [] node {b} (aab);
        \draw (aa) edge [bend left=10, above,red] node {\hspace{3pt}\red c} (aac);

        \draw (aaa) edge [] node {b} (aaab);

        \draw (aac) edge [bend right=10, above,red] node {\hspace{-4pt}\red b} (aacb);
        \draw (aac) edge [bend left=10, above,red] node {\hspace{4pt}\red c} (aacc);
        
        \draw (ac) edge [bend right=10, above] node {\hspace{-4pt}b} (acb);
        \draw (ac) edge [bend left=10, above] node {\hspace{4pt}c} (acc);
  
         \draw (acc) edge [] node {\hspace{4pt}a} (acca);
          
        \draw (aacc) edge [red] node {\red a} (aacca);
    
        \draw (aacca) edge [bend right=10, above] node {\hspace{-4pt}a} (aaccaa);
        
        \draw (aacca) edge [bend left=10, above] node {\hspace{4pt}c} (aaccac);
        
        \draw (aaccaa) edge [] node {b} (aaccaab);
        
        \draw (aacb) edge [dashed,orange,bend right=50] node {} (c);
        \draw (aa) edge [dashed,orange,bend right=80] node {} (aaa);
        \draw (aba) edge [dashed,orange,bend right=30] node {} (acca);
        \draw (acca) edge [dashed,orange,bend left=30] node {} (aacca);
        \draw (aacca) edge [dashed,orange,bend right=70] node {} (b);
        \draw (ab) edge [dashed,orange,bend right=30] node {} (aab);
        \draw (c) edge [dashed,orange,bend right=30] node {} (ac);
        \draw (aac) edge [dashed,orange,bend left=40] node {} (aaccac);
        \draw (abc) edge [dashed,orange,bend right=50] node {} (acc);
        \draw (aaa) edge [dashed,orange, bend right=70] node {} (aaccaa);
        \draw (abab) edge [dashed,orange, bend right=20] node {} (acb);
        \draw (root) edge [dashed,orange, bend right=40] node {} (a);

        \end{tikzpicture}\caption{Running example used throughout the paper. This repetitive trie has $n=26$ nodes (numbered in pre-order) and labels from the alphabet $\Sigma = $\{a,b,c\}. 
        The trie's topology, the colored nodes and the orange dashed edges are a concise representation of our compressed trie index discussed in Section \ref{sec:locate}. These components are discussed more in detail in the caption of Figure \ref{fig:example}.}\label{fig:example1}
    \end{figure}
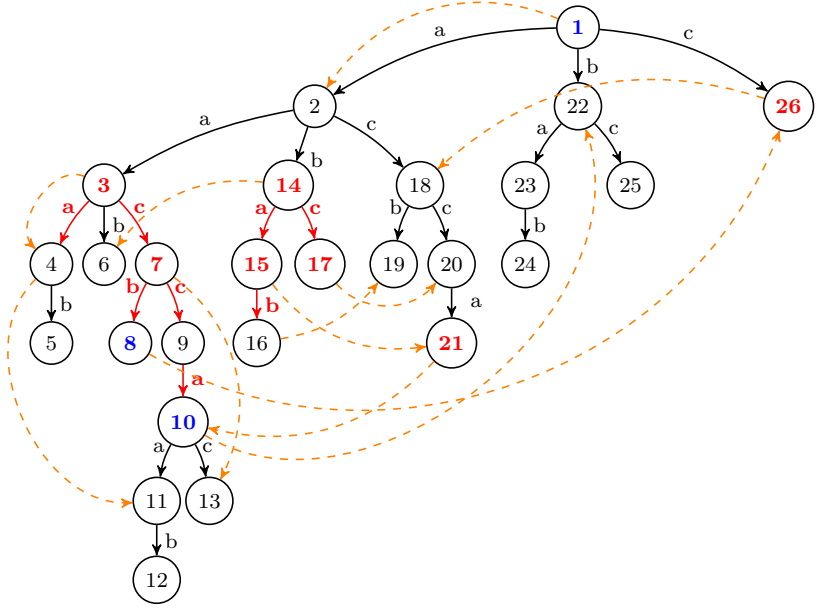
        
    \begin{figure}
        \centering
        \setlength{\tabcolsep}{2.5pt}
        \renewcommand{\arraystretch}{1.15}
        \begin{tabular}{|r|r|c|c|c|c|c|c|c|c|c|c|c|c|c|c|c|c|c|c|c|c|c|c|c|c|c|c|}\hline
            \multicolumn{2}{|c|}{$V$ (co-lex order)}&                 1 & 2 & 3 & 4 & 5 & 6 & 7 & 8 & 9 & 10 & 11&12&13&14&15&16&17&18&19&20&21&22&23&24&25&26\\\hline
            \multicolumn{2}{|c|}{$V$ (pre-order)}&                 \blue 1 & 2 & \red 3 & 4 & 11 & 23 & \red{15} & \red{21} & \blue{10} & 22 & \red{ 14}&6&5&12&24&16&19&\blue{8}&\red{26}&18&\red 7&13&25&\red{17}&20&9\\\hline
            \multicolumn{2}{|c|}{$\lambda$}&  $\#$&a&a&a&a&a&a&a&a&b&b&b&b&b&b&b&b&b&c&c&c&c&c&c&c&c \\\hline
            \multicolumn{2}{|c|}{}&                      a&a&\red a&&&&&&a&a&\red a&&&&&&&&&&&&&&a&\red a\\
            \multicolumn{2}{|c|}{XBWT ($out$)}&  b&b&b&b&b&b&\red b&&&&&&&&&&&&&b&\red b&&&&&\\
            \multicolumn{2}{|c|}{}&                      c&c&\red c&&&&&&c&c&\red c&&&&&&&&&c&\red c&&&&&\\ \hline
            & $DEL$ &               \multicolumn{3}{c|}{$\emptyset$}&\multicolumn{4}{c|}{\{a,c\}}&\{b\}&\multicolumn{3}{c|}{$\emptyset$}&\multicolumn{8}{c|}{\{a,c\}}&\multicolumn{2}{c|}{$\emptyset$}&\multicolumn{3}{c|}{\{b,c\}}&\multicolumn{2}{c|}{$\emptyset$}\\
            RL-XBWT       & $ADD$ & \multicolumn{3}{c|}{\{a,b,c\}}&\multicolumn{4}{c|}{$\emptyset$}&$\emptyset$&\multicolumn{3}{c|}{\{a,c\}}&\multicolumn{8}{c|}{$\emptyset$}&\multicolumn{2}{c|}{\{b,c\}}&\multicolumn{3}{c|}{$\emptyset$}&\multicolumn{2}{c|}{\{a\}}\\
            & $\ell$ &              \multicolumn{3}{c|}{3}&\multicolumn{4}{c|}{4}&1&\multicolumn{3}{c|}{3}&\multicolumn{8}{c|}{8}&\multicolumn{2}{c|}{2}&\multicolumn{3}{c|}{3}&\multicolumn{2}{c|}{2}\\\hline
            \multicolumn{2}{|c|}{$V/_{\equiv^r_<}$} & \multicolumn{3}{c|}{$\tilde{3}$}&\multicolumn{4}{c|}{$\tilde{15}$}&$\tilde{21}$&\multicolumn{3}{c|}{$\tilde{14}$}&\multicolumn{8}{c|}{$\tilde{26}$}&\multicolumn{2}{c|}{$\tilde{7}$}&\multicolumn{3}{c|}{$\tilde{17}$}&\multicolumn{2}{c|}{$\tilde{9}$}\\\hline
            \multicolumn{2}{|c|}{$V/_{\approx_<}$} & $\tilde{1}$&$\tilde{2}$&$\tilde{3}$&\multicolumn{4}{c|}{$\tilde{15}$}&$\tilde{21}$&\multicolumn{3}{c|}{$\tilde{14}$}&\multicolumn{8}{c|}{$\tilde{26}$}&$\tilde{18}$&$\tilde{7}$&\multicolumn{3}{c|}{$\tilde{17}$}&$\tilde{20}$&$\tilde{9}$\\\hline
            \multicolumn{2}{|c|}{$V/_{\equiv_<}$} & $\tilde{1}$&$\tilde{2}$&$\tilde{3}$&\multicolumn{4}{c|}{$\tilde{15}$}&$\tilde{21}$&$\tilde{10}$&\multicolumn{2}{c|}{$\tilde{14}$}&\multicolumn{7}{c|}{$\tilde{8}$}&$\tilde{26}$&$\tilde{18}$&$\tilde{7}$&\multicolumn{3}{c|}{$\tilde{17}$}&$\tilde{20}$&$\tilde{9}$\\\hline
        \end{tabular}
	\caption{
	\textbf{XBWT (Subsection \ref{sec:XBWT})}. First four rows: (1) the co-lex order and (2) the pre-order identifiers of the nodes of the tree in Figure \ref{fig:example1}, (3) the incoming label $\lambda(\preord u)$ of each node $\preord u\in V$, and (4) for each node, the characters labeling its outgoing edges.
	%(i.e. $out((\preord u)$. 
	Row (4) is the XBWT of the tree. 
	%Note that the $i$-th occurrence of $c\in\Sigma$ in the fourth (XBWT) row corresponds to the $i$-th occurrence of $c$ in the third ($\lambda$) row. For example, the 4-th occurrence of $a$ in the XBWT row corresponds to the edge labeled $a$ exiting node 10 and reaching node 11; indeed, the 4-th occurrence of $a$ in the $\lambda$ row corresponds to node 11. 
	\textbf{RL-XBWT (Subsection \ref{sec:XBWT})}. Fifth row: 
	the transform has $\nblocks=8$ blocks and $r=8$ runs. For each block, we store (1) the set $DEL$ of characters deleted w.r.t. the previous block (colored in red in the fourth row: these are the $c$-runs), (2) the set $ADD$ of characters added w.r.t. the previous block, and (3) the length $\ell$ (number of nodes) of the block. 
	%The run breaks are colored in red in the fourth row: the XBWT has $r = 8$ runs and $\nblocks = 8$ blocks. 
	%All the sets $DEL$ contribute with $7 \leq r$ characters. In the sets $ADD$, we have the first occurrence of every character plus at most one character for every run break, totaling $8 \leq \sigma + r \leq 2r$ characters. 
	\textbf{Wheeler Automata (Appendix \ref{sec:WA})}. In the last three rows of the table, we show the three quotients $V/_{\equiv^r_<}$ (RL-XBWT blocks), $V/_{\approx_<}$ (convex isomorphism), and $V/_{\equiv_<}$ (states of the minimum equivalent WDFA, by~\cite[Thm 4.1]{SODA20}). 
	%Note that each quotient is finer than the previous one, with $V/_{\equiv^r_<}$ being the coarsest.
	The notation $\tilde j$ indicates the equivalence class of node $\preord j$.
	%(where we chose as representative the maximum in the class by the Wheeler order). 
	\textbf{Tree attractors (Appendix \ref{sec:tree attr})}. 
	%In the tree, we color in red the edges corresponding to the run breaks (also colored in red) in the fourth row of the table. 
	By Theorem \ref{th:Gamma^r is tree attr}, the red edges (fourth row) form a tree attractor: every subtree has an isomorphic occurrence crossing a red edge. 
	%For example, the subtree induced by nodes 2,3,6,14,18 does not cross any red edge, but has an isomorphic occurrence rooted in node 3 that does. 	
	\textbf{Locate (Section \ref{sec:locate})}. A red node has different outgoing labels w.r.t. its co-lexicographic successor. A blue node has a different incoming label w.r.t. its co-lexicographic successor. Orange dashed arrows in Figure \ref{fig:example1} represent the sampled values of the co-lexicographic successor function: $\phi(\preord u_i) = \preord u_{i+1}$. 
	These arrows depart from red nodes and from nodes reached by red edges.
	}\label{fig:example}
\end{figure}

%Consider Definition \ref{def:XBWT}.
It is well known that the number of equal-letter runs in the Burrows-Wheeler Transform (BWT) of a string~\cite{burrows1994block} (that is, the XBWT of a simple labeled path) is highly correlated with the string's repetitiveness~\cite{MNSV08,kempa2019resolution}. 
As a result, the \emph{run-length encoded BWT} is a very powerful compressor for repetitive strings (see also~\cite{r-index}).
We now extend this technique to the XBWT of a trie and show that it enjoys many of the useful properties of the run-length encoded BWT. 

We say that $1\leq i < n$ is a $c$-run break, with $c\in \Sigma$, if $c\in out(\preord u_i)$ and either (i) $i=n$ or (ii) $c\notin out(\preord u_{i+1})$.
When $c$ is not specified, we simply say that $i$ is a \emph{run-break} (for some $c$).
%A $c$-run is a \emph{maximal} subsequence $i,i+1, \dots, i+k$ such that $i+k$ is a $c$-run break and  $i,i+1, \dots, i+k-1$ are not. 
Let $r_c(\mathcal T)$ be the number of $c$-run breaks.
%$$r_c(\mathcal T) = |\{ i\leq n\ :\ c\in out(\preord u_i) \wedge (i=n \vee c \notin out(\preord u_{i+1}))\}|.$$
We define the \emph{number $r(\mathcal T)$ of XBWT runs} as $r(\mathcal T) = \sum_{c\in\Sigma} r_c(\mathcal T)$. For brevity, in the following we will omit $\mathcal T$ and simply write $r_c$ and $r$.
The fourth row of Figure \ref{fig:example} shows run breaks in red. In the figure, we have $r=8$.
If $\mathcal T$ is a path (that is, a string), then $r$ coincides with the number of equal-letter runs in the BWT of $\mathcal T$.
% (in the variant where we co-lexicographically sort text prefixes). 

In the following definition we present the run-length (RL) encoded XBWT. See Figure \ref{fig:example} for a running example.
Importantly, note the distinction between XBWT \emph{runs} and \emph{blocks}.

\begin{definition}\label{def:RL-BWT}
The RL-XBWT of a trie $\mathcal T$ is the sequence of $\nblocks$ triples
$\langle (ADD_i,DEL_i,\ell_i) \rangle_{i=1}^{r'}$
obtained as follows. Break the sequence $\preord u_1, \preord u_2, \dots, \preord u_n$ into maximal contiguous \emph{blocks} 
such that the nodes in the same block $\preord u_{i}, \preord u_{i+1}, \dots, \preord u_{i+\ell-1}$ 
satisfy $out(\preord u_j) = out(\preord u_{j'})$, for all $i\leq j,j'< i+\ell$.
Only for the sake of this definition, let 
$out(\preord u_0) = \emptyset$.
The $q$-th block, starting with node $\preord u_{i_q}$ is then encoded with the triple $(ADD_q,DEL_q,\ell_q)$, where 
$
ADD_q = out(\preord u_{i_q}) - out(\preord u_{i_q-1})
$, 
$
DEL_q = out(\preord u_{i_q-1}) - out(\preord u_{i_q})
$, 
and $\ell_q$ is the length (number of nodes) of the block.
\end{definition}

The representation of Definition \ref{def:RL-BWT} is sufficient to reconstruct the XBWT:
$out(\preord u_{i_q}) = (out(\preord u_{i_q-1}) - DEL_{q}) \cup ADD_q$. In the next lemma we show that our representation can be stored in $O(r)$ space:

\begin{lemma}\label{lem:RL-BWT}
The RL-XBWT representation takes $O(r)$ words to be stored. 
\end{lemma}
\begin{proof}
 Let $A = \sum_{q=1}^{\nblocks} |ADD_q|$ and $D = \sum_{q=1}^{\nblocks} |DEL_q|$.
Note that the union of all sets $DEL_q$ contains the labels of all run breaks except the ones in the last position $n$: $D \leq r$ (see also Figure \ref{fig:example}).
The first occurrence of a character in $XBWT(\mathcal T)$ appears in the set $ADD_q$ of the corresponding block. These characters contribute $\sigma \leq r$ to the total size $A$ of these sets. Furthermore, each other element $c\in ADD_q$ is charged to the previous $c$-run break in the XBWT. It follows that $A \leq 2r$. Finally, note that $ADD_q \cup DEL_q \neq \emptyset$ must hold for every $q$, since otherwise the outgoing labels of the $q$-th block would coincide with those of the $(q-1)$-th block. Since $A+D \leq 3r$, this implies that there are also at most $\nblocks \leq 3r$ blocks. Our thesis follows. \qed
\end{proof}

\subsection{Relation with the Trie's Entropy}\label{sec:Hk}

A \emph{tree entropy} measure quantifies the amount of information in a labeled tree, either capturing the amount of predictability of its labels, its topology, or both.
Several notions of empirical entropy for trees have been considered in the literature so far. Ferragina et al.~\cite{XBWT} define the high-order empirical entropy of the tree's labels. This notion, however, does not take into account the tree's topology and is defined for arbitrary labeled trees. Jansson et al.~\cite{JANSSON2012619}, Hucke et al.~\cite{Hucke19}, and Ganczorz \cite{ganczorz2020using}  define tree entropy measures taking into account also the topology. Also their notions, however, work for arbitrary trees. 

The \emph{worst-case entropy $\mathcal C(n,\sigma)$ of a trie} considered by Raman et al.~\cite{RRR} is the measure we consider as starting point in this section. This quantity is defined as $\mathcal C(n,\sigma) = \log_2(|\mathcal U_{n,\sigma}|)$, where $\mathcal U_{n,\sigma}$ is the universe containing all tries with $n$ nodes on an alphabet of cardinality $\sigma$. 
Note that this is a clear lower bound (in bits) for encoding the trie, given only knowledge about $n$ and $\sigma$.
%(for simplicity, in this section we drop ceilings).
Measure $\mathcal C(n,\sigma)$ is still too weak for our purposes; we now show how to model also character frequencies (that is, zero-order compression) and, ultimately, high-order compression. 
Intuitively, our goal is to compute the information-theoretic lower bound for encoding the trie's labels given that we know the probability of seeing the label of an edge, conditioned on the path of length $k$ preceding it (for all $\sigma^k$ combinations of possible paths). 
%We now show how to strengthen $\mathcal C(n,\sigma)$ by taking into account also these aspects. 
On strings, it is well known that this notion of entropy has a strong relation with the notion of \emph{empirical entropy}~\cite{Kosaraju99}. For example, on binary alphabet the two measures differ at most by an additive $O(\log n)$ term~\cite{navarro2016compact}.

As in previous studies~\cite{XBWT,Hucke19}, we work in a model where the string $\pi_k[\preord u]$ of the last $k$ labels seen on the path connecting the root to a node $\preord u$ is a good predictor for the set $out(\preord u)$. 
More formally, $\pi_0[\preord u] = \epsilon$ (empty string), $\pi_1[\preord u] = \lambda(\preord u)$ and $\pi_k[\preord u] = \pi_{k-1}[\pi(\preord u)] \cdot \lambda(\preord u)$ for $k>1$. For this to be well-defined, we also set $\pi(1) = 1$ (1 is the root) to pad with $\lambda(1) = \#$ the contexts of nodes at depth less than $k$.
Let $X = X_1, \dots, X_{n'}$ be a sequence of ${n'}$ subsets of $\Sigma$ such there are $n'_c = \sum_{i=1}^{n'} |X_i\cap \{c\}|$ occurrences of character $c$ in the sequence, for all $c\in\Sigma$. The \emph{zero-order worst-case} entropy $\mathcal H^{wc}(X)$ of $X$ is defined as the logarithm of the size of the universe containing all set sequences of length $n'$ having the same characters' frequencies as $X$ (see Navarro~\cite{navarro2016compact}):
$
\mathcal H^{wc}(X) = \log_2 \left( \prod_{c\in\Sigma} {{n'}\choose{n'_c}} \right) = \sum_{c\in\Sigma} \log_2 {{n'}\choose{n_c'}}
$.
In the following we will simply write $\mathcal H^{wc}$ when the characters' frequencies are clear from the context. Note that $\mathcal H^{wc}(XBWT) \leq \mathcal C(n,\sigma)$, since the former fixes the frequencies of each character while the latter allows any frequency combination summing up to $n-1$.

At this point, we adapt the approach of Ferragina et al.~\cite{XBWT}. We define the sequence of sets
$cover(\rho) = \langle out(\preord u_{i}) \rangle_{i\ :\ \pi_k[\preord u_{i}] = \rho}$. Intuitively, $cover(\rho)$ is the sequence of sets containing all characters labeling edges that follow a path labeled with string $\rho$. The order by which the sets of $cover(\rho)$ are arranged is not important, as we apply zero-order compression to their elements. We define:

\begin{definition}
%[$k$-th order worst-case entropy of a trie $\mathcal T$]
\label{def:Hk}
$
\mathcal H^{wc}_k(\mathcal T) = \sum_{\rho \in \Sigma^{k}} \mathcal H^{wc}(cover(\rho)) 
$
\end{definition}
In the following we will simply write $\mathcal H^{wc}_k$ when $\mathcal T$ is clear from the context. Clearly, $\mathcal H^{wc}_k \leq \mathcal H^{wc}$ since $\mathcal H^{wc}_k$ fixes the characters' frequencies for each context $\rho$. 
%Quantity $\mathcal H^{wc}_k$ is a lower bound to the number of bits required to compress the trie's labels, given that we know the empirical probabilities of reading the character labeling an edge conditioned on the labels of the path of length $k$ leading to that edge.

The next step is to relate $r$ with $\mathcal H^{wc}_k$. On strings, it is well known that $r$ lower-bounds the $k$-th order empirical entropy~\cite{RLFM}. We show that this is the case also for the worst-case entropy of tries.

\begin{theorem}\label{th:r<nHk}
    The number $r$ of XBWT runs is always at most $\mathcal H^{wc}_k + \sigma^{k+1}$ for any $k\geq 0$.
\end{theorem}

In Appendices \ref{sec:WA} and \ref{sec:tree attr} we relate $r$ with other repetitiveness measures on tries: tree attractors \cite{attractors,prezzaOnStringAttractors} and the size of the smallest equivalent Wheeler automaton \cite{SODA20,GAGIE201767}.

%We will use Corollary \ref{cor:r<nHk} to express the size of our structures as a function of the trie's worst-case entropy. However, we stress out that
%, as opposed to entropy~\cite{KN13}, 
%measure $r$ captures large self-repetitions~\cite{MNSV08,kempa2019resolution} and is therefore a better measure for repetitive tries. 
%The empirical high-order entropy (per symbol) $H_k$ does not capture long repetitions: on a text $T$, the relation $H_k(T\cdot T) \geq H_k(T)$ always holds~\cite{KN13}. 
%In the next sections we compare $r$ with other repetition-aware measures on trees: the size of the smallest deterministic Wheeler automaton equivalent to the trie and the size of the smallest tree attractor. 

\section{Locating Paths in Compressed Tries}\label{sec:locate}

The \emph{locate problem} can be naturally generalized from strings to labeled trees as follows: given a pattern $P$, return the pre-order identifier $\preord u$ of all nodes such that $\lambda(1\rightsquigarrow \preord u)$ is suffixed by $P$. In such a case, we will say that $\preord u$ is \emph{reached by a path labeled $P$}. 
Plugging up-to-date data structures~\cite{optimal-rank-select} in the XBWT of Ferragina et al.~\cite{XBWT}, this structure takes $2n+o(n)$ bits on top of the entropy-compressed labels and \emph{counts} nodes reached by a path labeled with a pattern $P\in\Sigma^m$ in $O\left(m\log\log_w \sigma\right)$ time. 
We observe that it is straightforward to support also \emph{locate} queries on the XBWT by extending the standard solution (based on \emph{sampling}) used in compressed suffix arrays:
%(proof in Appendix \ref{app:locate XBWT}):

\begin{lemma}\label{lem:locate XBWT}
For any $1\leq t \leq n$, the XBWT can be augmented with additional $O((n/t)\log n) + o(n)$ bits so that, after counting, the pre-order identifiers of all $occ$ nodes reached by a path labeled with a pattern $P\in\Sigma^m$ can be returned in $O(occ\cdot t\log\log_w \sigma)$ time.
\end{lemma}

The simple solution of Lemma \ref{lem:locate XBWT} has 
the issue that
%two issues. First, it represents the succinct topology \emph{twice}. The reason is that the topology stored in the XBWT works in co-lexicographic order rather than pre-order, 
%and we need the latter to locate. This could be mitigated using sampling as done in \cite{ANSalgor10}: as the XBWT can be used to visit nodes in pre-order, we Second, 
the trade-off $t$ allows obtaining either a fast but large index or a slow and small index. 

The goal of this section is to solve both the above issues. More in detail, we show that a structure of $O(r) \subseteq O(\mathcal H^{wc}_k)$ words on top of a succinct topology representation of $2n+o(n)$ bits is sufficient to locate path occurrences in optimal constant time each. We start with navigation operations  that will be needed in our index.

\begin{enumerate}
	%\item \texttt{Child by label} $cbl(\wheel u,c)$.
	%Given a Wheeler-order node $\wheel u$ and a label $c$, return the child $\wheel v$ of $\wheel u$ reached by following the edge labeled with character $c$. \label{op:cbl}
	%\item \texttt{Parent} $par(\wheel u,c)$. 	Given a Wheeler-order node $\wheel u$, return the Wheeler-order parent of $\wheel u$. \label{op:par}
	\item \texttt{Child rank} $cr(\wheel u,c)$. Given the co-lex order $\wheel u=i$ of a node and a label $c\in out(\preord u_{i})$, return the integer $k$ such that the edge connecting $\wheel u$ with its $k$-th child is labeled with character $c$.\label{op:cr}
	\item \texttt{Depth} $depth(\preord u)$. Return the depth of pre-order node $\preord u$ (where the root has depth 0). \label{op:depth}
	\item \texttt{Child by rank} $cbr(\preord u,k)$. Return the $k$-th (pre-order) child of pre-order node $\preord u$.\label{op:cbr}
	%\item \texttt{Parent} $\pi(\preord u)$. Return the parent of pre-order node $\preord u$,
	\item \texttt{Sibling rank} $sr(\preord u)$. Return the integer $k$ such that $\preord u$ is the $k$-th child of its parent.\label{op:sr}
	\item \texttt{Lowest Common Ancestor} $LCA(\preord u, \preord v)$ of two pre-order nodes $\preord u$ and $\preord v$.\label{op:lca}
	\item \texttt{Level Ancestor Queries} $LAQ(\preord u, \ell)$. Given $\ell\geq 1$, return $\pi^{(\ell)}(\preord u)$, that is, the parent function $\pi$ applied $\ell$ times to pre-order node $\preord u$.\label{op:lac}
	\item \texttt{Isomorphic Descendant} $ISD(\preord u, \preord v, \preord u')$. Let $\preord v$ be a descendant of $\preord u$ reached by following a path $\preord u \rightarrow \preord w \rightsquigarrow \preord v$ with $\alpha = \lambda(\preord w \rightsquigarrow \preord v)$, and let $\preord u' \approx \preord u$ be a node isomorphic to $\preord u$. This operation returns the descendant $\preord v'$ of $\preord u'$ reached by following the path $\preord u' \rightarrow \preord w' \rightsquigarrow \preord v'$ with $\lambda(\preord w' \rightsquigarrow \preord v') = \alpha$.\label{op:isd}
	\item \texttt{Isomorphic Child} $ISC(\preord u_i, k)$. Given a pre-order node $\preord u_i$, $i<n$, such that $out(\preord u_i) \neq out(\preord u_{i+1})$ and given an integer $1\leq k \leq |out(\preord u_i)|$, let $c = \lambda(cbr(\preord u_i,k))$ be the $k$-th smallest label in $out(\preord u_i)$. 
	Assuming that $c\in out(u_{i+1})$, this function returns the integer $t$ such that
	$c = \lambda(cbr(\preord u_{i+1},t))$. \label{op:isc}
\end{enumerate}

\begin{lemma}\label{lem:cr}
There is a data structure taking $O(r\log n) + o(n)$ bits of space and supporting operation $cr(\wheel u,c)$ in $O(\log\sigma)$ time. 
\end{lemma}

For the remaining operations, we store explicitly the topology. 
%Since children are sorted by lexicographically-increasing labels note that, crucially, none of the operations \ref{op:depth}-\ref{op:isc} requires access to the tree's labels (just to the topology, plus some small additional information). This is true even for operations \ref{op:isd} and \ref{op:isc}. In Operation \ref{op:isd}, it is sufficient to descend the path labeled $\alpha$ by children ranks (w.r.t. their siblings), rather than by children labels. Such path is uniquely determined by the children ranks, even if the labels are unknown. In operation \ref{op:isc}, it is sufficient to record which outgoing labels are shared between $\preord u_i$ and $\preord u_{i+1}$. We will show how to store this information using small extra space. 
Navarro and Sadakane~\cite{suctrees} show how to support operations \ref{op:depth}-\ref{op:lac} in $O(1)$ time using $2n + o(n)$ bits of space. 
%We are not aware of compressed tree representations supporting those operations, therefore in the rest of the paper we will stick to this representation. 
We show:

\begin{lemma}\label{lemma:isd}
    The structure of Navarro and Sadakane~\cite{suctrees} supports also $ISD(\preord u, \preord v,\preord u')$ in $O(1)$ time. 
\end{lemma}

\begin{lemma}\label{lemma:isc}
    Operation $ISC(\preord u_i, k)$ can be supported in $O(1)$ time and $O(r\log n) + o(n)$ bits of space.
\end{lemma}

%\subsection{Locating Pattern Occurrences}

%We are ready to show how to implement efficient \emph{locate} queries on the XBWT.
%: given a pattern $P\in\Sigma^m$, return the pre-order representation of the nodes reached by a path labeled $P$. 
Our strategy for supporting efficient \emph{locate} queries on the XBWT is a nontrivial generalization to tries of the \emph{r-index} data structure~\cite{r-index} (a locate machinery on strings).
Let $[\wheel \ell,\wheel r]$ be the co-lexicographic range of nodes reached by a path labeled $P$. 
We divide the problem of answering locate queries into two sub-problems. (1) \textbf{Toehold}: compute $[\wheel \ell, \wheel r]$ and $\preord u_{\wheel \ell}$. (2) \textbf{Climb}: evaluate function $\phi(\preord u_i) = \preord u_{i+1}$ for any $i<n$. The combination of (1) and (2) yields $\preord u_{\wheel \ell}, \dots, \preord u_{\wheel r}$.

The \emph{Toehold} step requires navigating the tree using both node representations. We show:

\begin{lemma}\label{lem:toehold}
There is a data structure taking $O(r\log n) + o(n)$ bits of space on top of the succinct tree topology of Navarro and Sadakane~\cite{suctrees} that, 
given a pattern $P\in \Sigma^m$, returns the co-lexicographic range $[\wheel \ell,\wheel r]$ of nodes reached by a path labeled $P$, as well as $\preord u_{\wheel\ell}$, in $O\left(m\log\sigma\right)$ time.
\end{lemma}

%We now show how to evaluate function $\phi(\preord u_i) = \preord u_{i+1}$ for any $i < n$ by performing a \emph{constant} number of jumps (from $\preord u_{i}$ to $\preord u_{i+1}$) on the tree topology via operations \ref{op:depth}-\ref{op:isc} and via a small ($O(r)$) number of \emph{explicit links} stored on the tree. 
%Let us move to the \emph{Climb} step.

We remark that a simple sampling of the nodes' pre-order identifiers is not sufficient to replace the succinct tree topology in Lemma \ref{lem:toehold}. The problem with this strategy is that, in absence of the explicit topology, it is not possible to navigate from sampled nodes to non-sampled ones, thus retrieving the pre-order identifier of the latter. Similarly, the succinct topology is fundamental to support operations 2-8: without the topology it seems challenging to navigate between the pre-order identifiers of the nodes. We leave it as an exciting open question whether it is possible to represent the topology in $O(r)$ words while supporting all operations in poly-logarithmic time. 

We now show how to implement the \emph{Climb} step with a \emph{constant} number of jumps (from $\preord u_{i}$ to $\preord u_{i+1}$) on the tree, each taking constant time.
We mark nodes in blue, red, or both (colors are not exclusive). 
A node $\preord u_{i}$, $i<n$, is \emph{red} if it does not have the same \emph{outgoing} labels as its co-lexicographic successor: $out(\preord u_{i}) \neq out(\preord u_{i+1})$. A node $\preord u_{i}$, $i<n$, is \emph{blue} if it does not have the same \emph{incoming} label as its co-lexicographic successor:  $\lambda(\preord u_{i}) \neq \lambda(\preord u_{i+1})$.
Since $r'\leq 3r$ (Lemma \ref{lem:RL-BWT}) and $\sigma \leq r$, there are $O(r)$ marked nodes in total.
Our running example in Figures \ref{fig:example1} and \ref{fig:example} shows how nodes are colored according to the above definitions. 

The following lemma shows that co-lexicographic adjacency is preserved when following equally-labeled edge pairs and is important for our construction. 

\begin{lemma}\label{lem:prec}
If $\preord u <_{pred} \preord v$ then $child_c(\preord u) <_{pred} child_c(\preord v)$ for all $c\in out(\preord u) \cap out(\preord v)$. 
\end{lemma}
\begin{proof}
Let $c\in out(\preord u) \cap out(\preord v)$, $\preord u' = child_c(\preord u)$, and $\preord v' = child_c(\preord v)$.
Suppose, by contradiction, that there exists $\preord w$ such that $\preord u' < \preord w < \preord v'$. By co-lexicographic Axiom (i) (see the beginning of Section \ref{sec:XBWT}), it must be the case that $c = \lambda(\preord u') = \lambda(\preord v') = \lambda(\preord w)$. Then, we have two cases. (a) $\pi(\preord w) < \preord u <_{pred} \preord v$, which by co-lexicographic Axiom (ii) implies $\preord w < \preord u'$, a contradiction. (b) $\preord u <_{pred} \preord v < \pi(\preord w)$, which by co-lexicographic Axiom (ii) implies $\preord v' < \preord w$, a contradiction. \qed
\end{proof}

Let $i<n$. By recursively applying Lemma \ref{lem:prec} to the descendants of a node, one can easily see the following:

\begin{corollary}\label{cor:no color -> isomorphic subtrees}
 $\preord u_i \not\approx \preord u_{i+1}$ if and only if the complete subtree rooted in $\preord u_i$ contains a red node.
\end{corollary}
\begin{proof}
 Assume that the complete subtree rooted in $\preord u_i$ does not contain any red node.
Since $\preord u_i$ is not red and $i<n$, then (by definition of red node) $out(\preord u_i) = out(\preord u_{i+1})$. But then, by Lemma \ref{lem:prec} $child_c(\preord u_i) <_{pred} child_c(\preord u_{i+1})$ for all $c \in out(\preord u_i) = out(\preord u_{i+1})$. The reasoning can be repeated inductively to the children of $\preord u_i$ until reaching the leaves, since the complete subtree rooted in $\preord u_i$ does not contain any red node. As a consequence, we obtain $\preord u_i \approx \preord u_{i+1}$.

Conversely, assume that the complete subtree rooted in $\preord u_i$ contains a red node. 
If $\preord u_i$ is red, then (by definition of red node) $out(\preord u_i) \neq out(\preord u_{i+1})$ and therefore $\preord u_i \not\approx \preord u_{i+1}$.
Otherwise, $\preord u_i$ is not red and we can repeat the reasoning to the children of $\preord u_i$ and $\preord u_{i+1}$ (as seen above). Since the complete subtree rooted in $\preord u_i$ contains a red node, at some point we will find a red descendant $\preord u_j$ of $\preord u_i$ such that $out(\preord u_j) \neq out(\preord u_{j+1})$, where $\preord u_{j+1}$ is the corresponding descendant of $\preord u_{i+1}$. As a consequence, $\preord u_i \not\approx \preord u_{i+1}$. \qed
\end{proof}

The following lemma shows that we can find colored descendants and ancestors in $O(1)$ time:

\begin{lemma}\label{lem:find colored}
 There is a data structure taking $O(r\log n) + o(n)$ bits of space on top of the succinct tree topology of Navarro and Sadakane~\cite{suctrees} and answering the following queries in $O(1)$ time. Given a pre-order node $\preord u_i$ with $i<n$:
 \begin{itemize}
     \item[(a)] 
     If $\preord u_i$ is not colored, find a colored node $\preord u_j \neq \preord u_i$ in the complete subtree rooted in $\preord u_i$ such that no node on the path from $\preord u_i$ to $\preord u_j$ is colored (except $\preord u_j$), or report that $\preord u_j$ does not exist.
     \item[(b)] Find the lowest ancestor $\preord u_j$ of $\preord u_i$ such that the complete subtree rooted in $\preord u_j$ contains a colored node. Note that such a node always exists, since the root is always blue. 
 \end{itemize}
\end{lemma}

%As it turns out, our coloring scheme has a lot of structure:

\begin{lemma}\label{lem:uj red}
In Lemma \ref{lem:find colored} (a), if $\preord u_j$ exists then $\preord u_j$ must be red and not blue.
\end{lemma}

We introduce the notion of \emph{adjacent paths}:

\begin{definition}\label{def:adj paths}
 We say that two paths $\preord u_{i_1} \rightarrow \preord u_{i_2} \rightsquigarrow \preord u_{i_k}$ and $\preord u_{j_1} \rightarrow \preord u_{j_2} \rightsquigarrow \preord u_{j_k}$ of the same length $k$ are \emph{adjacent} if it holds that $j_{t} = i_{t}+1$ for all $1\leq t \leq k$.
\end{definition}

%In Appendix \ref{app:cor:adj paths} we prove that uncolored paths have always an adjacent path:

\begin{lemma}\label{lem:adj paths}
Let $\Pi = \preord u_{i_1} \rightarrow \preord u_{i_2} \rightsquigarrow \preord u_{i_k}$, with $i_j<n$ for some $1\leq j \leq k$, be a path of length $k$ without blue nodes other than (possibly) $\preord u_{i_1}$ and without red nodes other than (possibly) $\preord u_{i_k}$. Then, $\preord u_{i_1+1} \rightarrow \preord u_{i_2+1} \rightsquigarrow \preord u_{i_k+1}$ is a path in the tree (adjacent to $\Pi$). 
\end{lemma}

We furthermore explicitly store (sample) the value of function $\phi$ on the following nodes: (1) on each colored node $\preord u_i$, we explicitly store $\phi(\preord u_i) = \preord u_{i+1}$. We call these \emph{$\phi$-samples of type 1}. (2) Let $\preord u_i$, $i<n$, be such that $c\in out(\preord u_{i})$ and $c\notin out(\preord u_{i+1})$. Let moreover $\preord u_{j} = child_c(\preord u_{i})$. If $j<n$, then we explicitly store $\phi(\preord u_j) = \preord u_{j+1}$ on node $\preord u_{j}$. We call these \emph{$\phi$-samples of type 2}. 
Note that a $\phi$-sample could be both of type 1 and 2 (for example, see Figure \ref{fig:example1}, node 7). 
Since samples of type 1 are stored only on colored nodes and samples of type 2 correspond to run breaks, in total we explicitly store $O(r)$ $\phi$-samples. Figure \ref{fig:example1} depicts these samples as orange dashed arrows.
The color(s) and $\phi$-sample associated with colored/$\phi$-sampled pre-order nodes can be retrieved in constant time and $O(r\log n) + o(n)$ bits of space using an entropy-compressed bitvector~\cite{RRR} marking such nodes.

We are now ready to show how to compute $\phi(\preord u_i) = \preord u_{i+1}$ for any $1\leq i < n$.
We break our algorithm into cases. 
In Appendix \ref{app:examples climb} we discuss examples of all cases based on the trie of Figure \ref{fig:example1}.
%The  main two cases distinguish whether the complete subtree rooted in $\preord u_i$ contains a colored node.

\paragraph{\textbf{Case 1}: the complete subtree rooted in $\preord u_i$ contains colored nodes.} 
See Figure \ref{fig:cases1-2.1} (left).
If $\preord u_i$ is colored (red, blue, or both), then $\phi(\preord u_i)$ is explicitly stored.
Otherwise, we use Lemma \ref{lem:find colored} (a) to find a colored node $\preord u_j \neq \preord u_i$ in the complete subtree rooted in $\preord u_i$ such that no  node other than $\preord u_j$ on the path $\Pi = \preord u_i \rightsquigarrow \preord u_j$ is colored.
In particular, by Lemma \ref{lem:uj red} node $\preord u_j$ must be red and no node on the path is blue.
Since $\Pi$ enjoys this property and $i<n$, we can apply
Lemma \ref{lem:adj paths} to it and obtain that $\preord u_{i+1} \rightsquigarrow \preord u_{j+1}$ is a valid path with $t = depth(\preord u_j) - depth(\preord u_i)$ edges and it is adjacent to $\Pi$. 
%Note moreover that, since $\pi(\preord u_j)$ is not colored, it must be the case that $\lambda(\preord u_j) = \lambda(\preord u_{j+1})$, therefore $\preord u_j$ cannot be blue and it must be red.
We find $\preord u_{j+1} = \phi(\preord u_j)$, which is stored explicitly since $\preord u_j$ is red.
Finally, we jump to $\preord u_{i+1}$ with a level ancestor query by $t$ levels from $\preord u_{j+1}$. More formally, we obtain:
$$
\phi(\preord u_{i}) = \preord u_{i+1} =  LAQ(\phi(\preord u_j), depth(\preord u_j) - depth(\preord u_i)).
$$
We note that on tries being simple paths (i.e. strings) it is always the case that $\preord u_i$ has a red descendant (that is, the unique leaf).
It follows that the above equation is always applied when the tree is a string. In fact, in this case the equation reduces to what is implemented in the \emph{r-index} data structure~\cite{r-index}. On trees, however, things are more complicated: it is not always the case that $\preord u_i$ has colored descendants. This case is treated below.

\begin{figure}
	\vspace{-10pt}
	\begin{minipage}{0.44\textwidth}
		\centering
		\includegraphics[width=0.95\textwidth]{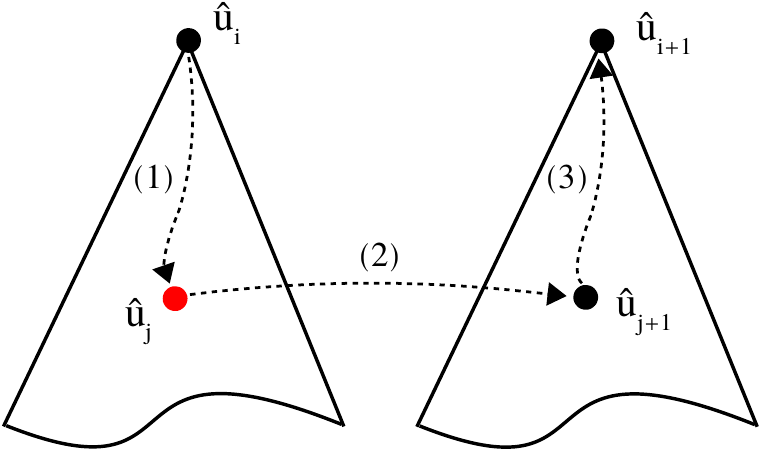}
	\end{minipage}
	\begin{minipage}{0.1\textwidth}
		\  
	\end{minipage}
	\begin{minipage}{0.44\textwidth}
		\centering
		\includegraphics[width=\textwidth]{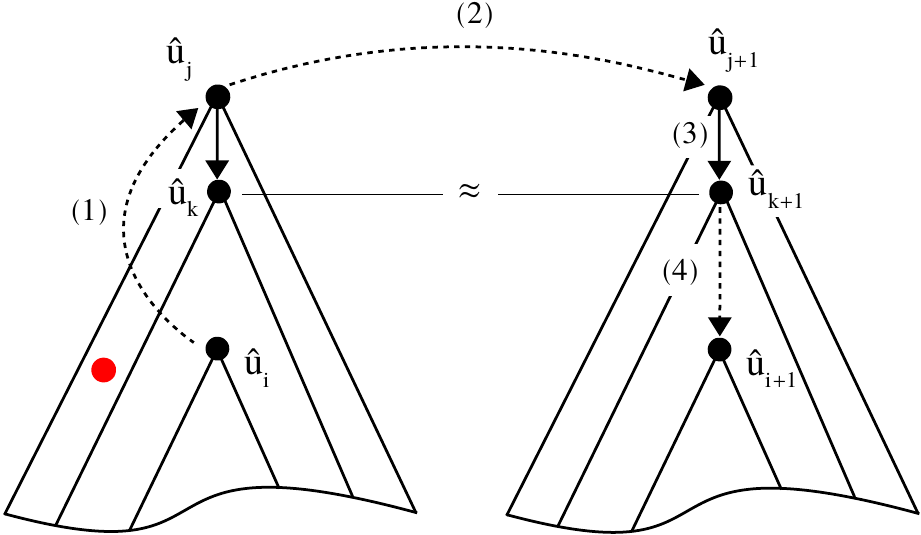}
	\end{minipage}
	\caption{\textbf{Left}. Locate, \emph{Case 1}: node $\preord u_i$ has colored descendants. (1) Find colored descendant $\preord u_j$. (2) Follow $\phi$-sample. (3) Level Ancestor Query. \textbf{Right}. Locate, \emph{Case 2.1}: node $\preord u_i$ does not have colored descendants, $\preord u_j$ is the lowest ancestor that does, and $\preord u_j$ is not red. (1) Find $\preord u_j$. (2) Apply \emph{Case 1} to find $\preord u_{j+1}$. (3) Descend to child $\preord u_{k+1}$, and note that $\preord u_{k} \approx \preord u_{k+1}$. (4) Find $\preord u_{i+1}$ with an \emph{isomorphic descendant} query.}\label{fig:cases1-2.1}
\end{figure}

\paragraph{\textbf{Case 2}: the complete subtree rooted in $\preord u_i$ does not contain any colored node.} 
See Figures \ref{fig:cases1-2.1} (right) and \ref{fig:cases2.2.1-2.2.2}.
The idea is to navigate upwards instead of downwards as done in Case 1. 
We first find, using Lemma \ref{lem:find colored} (b), the lowest ancestor $\preord u_j$ of $\preord u_i$ such that the complete subtree rooted in $\preord u_j$ contains a colored node. 
We further distinguish two sub-cases, depending on whether $\preord u_j$ is red or not.

\paragraph{\textbf{Case 2.1}: $\preord u_j$ is not red.} 

See Figure \ref{fig:cases1-2.1} (right).
Consider the path $\preord u_j \rightarrow \preord u_{k} \rightsquigarrow \preord u_i$, where $\preord u_{k}$ is child of $\preord u_j$ on the path. 
Index $k$ might coincide with $i$; in this case, the path is simply $\preord u_j \rightarrow \preord u_i$.
Note that $i<n$ and that no node $\preord v$ in this path is colored except, possibly, $\preord u_j$ (which might be blue), otherwise we would have chosen $\preord v$ in place of $\preord u_j$. 
Then, we can apply Lemma \ref{lem:adj paths} and obtain that $\preord u_{j+1} \rightarrow \preord u_{k+1} \rightsquigarrow \preord u_{i+1}$ must also be a path. In particular, $j<n$.

Next, we show that we can retrieve $\phi(\preord u_j) = \preord u_{j+1}$ in constant time. Since $\preord u_j$ is either not colored or blue, then (by definition of $\preord u_j$) there must be a colored node (possibly $\preord u_j$ itself) in the subtree rooted in $\preord u_j$. Then, since $j<n$ we can apply \emph{Case 1} and find $\phi(\preord u_j) = \preord u_{j+1}$ in constant time. 

Let $t = depth(\preord u_i) - depth(\preord u_j)$. We compute $\preord u_k = LAQ(\preord u_i, t-1)$ and get its rank $q$ among its siblings with $q = sr(\preord u_k)$. Since $\preord u_j$ is not red, we have that 
%$\preord u_j \equiv^r \preord u_{j+1}$: 
the two nodes have the same outgoing labels. 
Since $\preord u_k$ is not blue, we have $\lambda(\preord u_k) = \lambda(\preord u_{k+1})$. These two observations imply that $\preord u_{k+1}$ is the $q$-th children of $\preord u_{j+1}$ as well: we compute it as $\preord u_{k+1} = cbr(\preord u_{j+1},q)$. 
If $t=1$ then $\preord u_{k+1}$ coincides with $\preord u_{i+1}$ and we are done. Otherwise, since $\preord u_j$ is the lowest ancestor of $\preord u_i$ such that the complete subtree rooted in $\preord u_j$ contains a colored node, the subtree rooted in $\preord u_{k}$ does not contain any colored node. By Corollary \ref{cor:no color -> isomorphic subtrees}, we obtain $\preord u_{k} \approx \preord u_{k+1}$: the two complete subtrees are isomorphic. But then, we can finally find $\preord u_{i+1}$ with an isomorphic descendant query: $\preord u_{i+1} = ISD(\preord u_{k}, \preord u_{i}, \preord u_{k+1})$.

\paragraph{\textbf{Case 2.2}: $\preord u_j$ is red.}
See Figure \ref{fig:cases2.2.1-2.2.2}.
Consider the path $\preord u_j \rightarrow \preord u_{k} \rightsquigarrow \preord u_i$, where $\preord u_{k}$ is child of $\preord u_j$ on the path ($k$ might coincide with $i$; in this case, the path is simply $\preord u_j \rightarrow \preord u_i$). 
%By definition of $\preord u_j$, no node in the sub-path $\preord u_{k} \rightsquigarrow \preord u_i$ is colored. Then, by Lemma \ref{lem:adj paths} we obtain that this sub-path is adjacent to the path $\preord u_{k+1} \rightsquigarrow \preord u_{i+1}$.
Let $t = depth(\preord u_i) - depth(\preord u_j)$. We find $\preord u_k = LAQ(\preord u_i,t-1)$.
We distinguish two sub-cases.

\begin{figure}
	\begin{minipage}{0.44\textwidth}
		\centering
		\includegraphics[width=\textwidth]{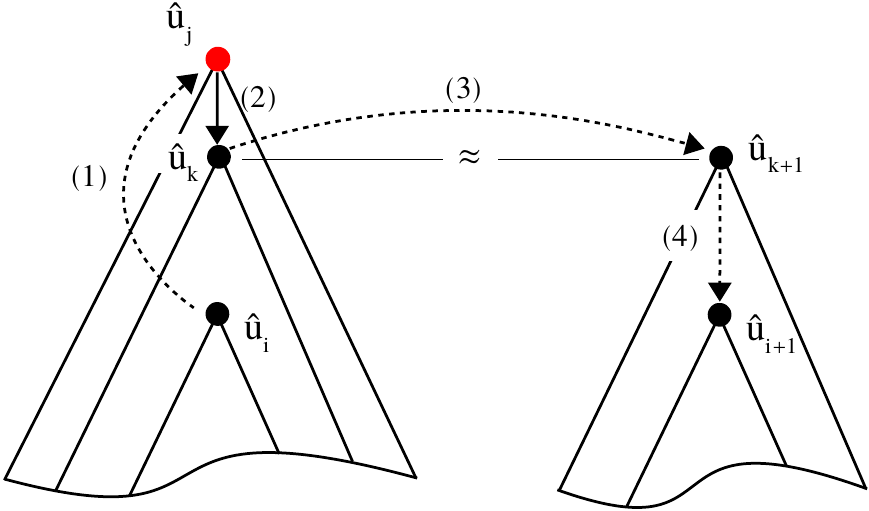}
	\end{minipage}
	\begin{minipage}{0.1\textwidth}
		\  
	\end{minipage}
	\begin{minipage}{0.44\textwidth}
		\centering
		\includegraphics[width=\textwidth]{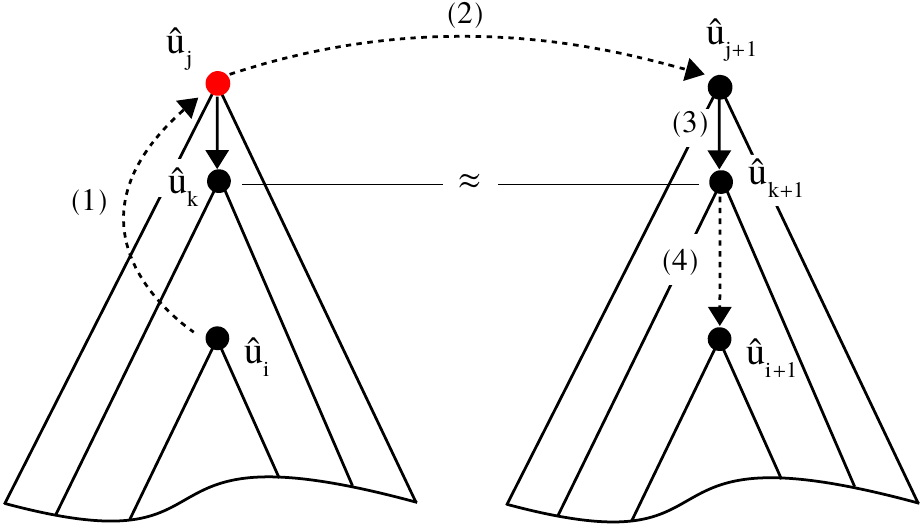}
	\end{minipage}
	\caption{\textbf{Left}. Locate, \emph{Case 2.2.1}: 
		node $\preord u_i$  does not have colored descendants, $\preord u_j$ is the lowest ancestor that does, $\preord u_j$ is red, and $\phi(\preord u_k)$ is a $\phi$-sample. (1) Find $\preord u_j$. (2) Descend to $\preord u_k$. (3) Follow $\phi$-sample and note that $\preord u_{k} \approx \preord u_{k+1}$. (4) Find $\preord u_{i+1}$ with an \emph{isomorphic descendant} query. \textbf{Right}. Locate, \emph{Case 2.2.2}: 
		node $\preord u_i$  does not have colored descendants, $\preord u_j$ is the lowest ancestor that does, $\preord u_j$ is red, and $\phi(\preord u_k)$ is not a $\phi$-sample. (1) Find $\preord u_j$. (2) Follow $\phi$-sample. (3) Descend to child $\preord u_{k+1}$ using an \emph{isomorphic child} query and note that $\preord u_{k} \approx \preord u_{k+1}$. (4) Find $\preord u_{i+1}$ with an \emph{isomorphic descendant} query. }\label{fig:cases2.2.1-2.2.2}
\end{figure}

\paragraph{\textbf{Case 2.2.1}: $\phi(\preord u_k)$ is a $\phi$-sample of type 2.} 
%See Figure \ref{fig:cases2.2.1-2.2.2}.
%Let $c = \lambda(\preord u_{k})$. Then, there must exist a node $\preord v$ such that $\preord u_{j} < \preord v <  \preord u_{j'}$ and $c\notin out(\preord v)$.
See Figure \ref{fig:cases2.2.1-2.2.2} (left).
Then, we retrieve $\preord u_{k+1} = \phi(\preord u_k)$ in constant time. Since $\preord u_j$ is the lowest ancestor of $\preord u_i$ such that the complete subtree rooted in $\preord u_j$ contains a colored node, the subtree rooted in $\preord u_{k}$ does not contain any colored node. By Corollary \ref{cor:no color -> isomorphic subtrees}, this implies that $\preord u_{k} \approx \preord u_{k+1}$: the two complete subtrees are isomorphic. But then, we can  find $\preord u_{i+1}$ with an isomorphic descendant query: $\preord u_{i+1} = ISD(\preord u_{k}, \preord u_{i}, \preord u_{k+1})$. 
%Even if it is not important for the above reasoning, it can easily be shown that the ancestor $\preord u_{j'}$ of $\preord u_{k+1}$ is such that $j'>j+1$.

\paragraph{\textbf{Case 2.2.2}:
$\phi(\preord u_k)$ is not a $\phi$-sample of type 2.} 
See Figure \ref{fig:cases2.2.1-2.2.2} (right).
Since $i<n$ and $\preord u_{k} \rightsquigarrow \preord u_{i}$ does not contain colored nodes then by Lemma \ref{lem:adj paths} $\preord u_{k+1} \rightsquigarrow \preord u_{i+1}$ is a path in the tree.
Since $\preord u_{k}$ is not blue, we have that $\lambda(\preord u_{k}) = \lambda(\preord u_{k+1})$. In particular, $\lambda(\preord u_{k+1}) \neq \#$ (because $\lambda(\preord v)=\#$ only for $\preord v = 1$) so $\preord u_{k+1}$ is not the root.
Let $\preord u_{j'}$ be the parent of $\preord u_{k+1}$.
By co-lexicographic Axiom (ii) (see the beginning of Section \ref{sec:XBWT}), $\lambda(\preord u_{k}) = \lambda(\preord u_{k+1})$ and $\preord u_{k} < \preord u_{k+1}$ imply  $\preord u_{j} < \preord u_{j'}$ (otherwise, Axiom (ii) would force $\preord u_{k+1} < \preord u_{k}$, a contradiction). 
We can say more: since $\phi(\preord u_k)$ is not a $\phi$-sample of type 2, then it must be the case that $\preord u_{j} <_{pred} \preord u_{j'} = \preord u_{j+1}$. Assume, for contradiction, that this were not true, i.e. that there existed a node $\preord v$ such that $\preord u_{j} <_{pred} \preord v < \preord u_{j'}$. 
Let $c = \lambda(\preord u_k)$.
The cases are two: (a) $c \in out(\preord v)$. Then, by co-lexicographic Axiom (ii) it must be the case that $\preord u_k < child_c(\preord v) < \preord u_{k+1}$, a contradiction. (b) $c \notin out(\preord v)$. Then, $j$ would be a $c$-run break and $\phi(\preord u_k)$ would be a $\phi$-sample of type 2, a contradiction. 

%Since $i<n$ and since none of the nodes in $\preord u_{k} \rightsquigarrow \preord u_i$ are colored, Lemma \ref{lem:adj paths} and $j'=j+1$ imply that $\preord u_j \rightarrow \preord u_{k} \rightsquigarrow \preord u_i$ and $\preord u_{j+1} \rightarrow \preord u_{k+1} \rightsquigarrow \preord u_{i+1}$ are (adjacent) paths. 
%In particular, note that $\lambda(\preord u_k) = \lambda(\preord u_{k+1})$ because $\preord u_k$ is not blue.
Since $\preord u_{j}$ is red, $\phi(\preord u_{j}) = \preord u_{j+1}$ is a $\phi$-sample of type 1 and we can retrieve it in constant time. 
Let $q = sr(\preord u_k)$: node $\preord u_k$ is the $q$-th among the children of its parent $\preord u_j$.
Since $\preord u_{j}$ is red, then $out(\preord u_{j}) \neq out(\preord u_{j+1})$.
This and the fact that $\lambda(\preord u_k) = \lambda(\preord u_{k+1})$ imply that we can find $\preord u_{k+1}$ with an isomorphic child operation (Operation \ref{op:isc}):  
$\preord u_{k+1} = cbr(\preord u_{j+1},ISC(\preord u_j,q))$.

Since $\preord u_j$ is the lowest ancestor of $\preord u_i$ such that the complete subtree rooted in $\preord u_j$ contains a colored node, the subtree rooted in $\preord u_{k}$ does not contain any colored node. By Corollary \ref{cor:no color -> isomorphic subtrees}, this implies that $\preord u_{k} \approx \preord u_{k+1}$: the two complete subtrees are isomorphic. But then, we can finally find $\preord u_{i+1}$ with an isomorphic descendant query: $\preord u_{i+1} = ISD(\preord u_{k}, \preord u_{i}, \preord u_{k+1})$.

We obtain our final result: 

\begin{theorem}\label{th:locate}
	Let $\mathcal T$ be a trie with $n$ nodes whose XBWT has $r$ runs.
	Our index takes $2n + o(n) + O(r\log n)$ bits of space and locates the pre-order identifiers of the $occ$ nodes reached by a path labeled with  $P\in\Sigma^m$ in $O\left(m\log\sigma + occ\right)$ time. 
\end{theorem}

In Appendix \ref{app:entropy bound} we bound the size of our index as a function of the trie's worst-case entropy $\mathcal H^{wc}_k$.
Note that the whole locate machinery, as well as the edges' labels, fits within compressed space on top of the succinct topology. Moreover, the topology is stored using Navarro and Sadakane's representation~\cite{suctrees}, which supports much more advanced navigation queries than the XBWT~\cite{XBWT}. We note that improvements in navigation queries \ref{op:depth}-\ref{op:isc} on compressed trees will have a direct impact on our index. We leave it as an exciting open question whether it is possible to support those queries within $O(r)$ words of space, thus reducing the size of our index to $O(r)$ words \emph{in total}.

\newpage

\appendix

\section{Proof of Lemma \ref{lem:locate XBWT}}\label{app:locate XBWT}

\paragraph{\textbf{Claim}.} For any $1\leq t \leq n$, the XBWT can be augmented with additional $O((n/t)\log n) + o(n)$ bits so that, after counting, the pre-order identifiers of all $occ$ nodes reached by a path labeled with a pattern $P\in\Sigma^m$ can be returned in $O(occ\cdot t\log\log_w \sigma)$ time.
\paragraph{\textbf{Proof}.} 

We exploit the fundamental property (used also in count queries) that characters occur in the same relative order in XBWT and in the sequence $\Lambda = \lambda(\preord u_1), \dots, \lambda(\preord u_n)$ (see also Figure \ref{fig:example}): the $i$-th occurrence of character $c\in\Sigma$ in the sequence $XBWT=out(\preord u_1), \dots, out(\preord u_n)$ corresponds to the same edge associated with the $i$-th occurrence of character $c\in\Sigma$ in the sequence $\Lambda$. 
Using up-to-date rank and select data structures~\cite{optimal-rank-select},
this property allows performing local navigation operations (parent, children): to move to the parent of co-lex node $\wheel u = j$, count the number $q$ of occurrences of $\lambda(\preord u_j)$ occurring in $\lambda(\preord u_1), \dots, \lambda(\preord u_j)$ (one constant-time rank query on the bitvector representing $\Lambda$, see~\cite{XBWT}), and jump to the $q$-th occurrence of $\lambda(\preord u_j)$ in XBWT (one constant-time select query using the structures of~\cite{optimal-rank-select}). 
Using the inverse operation (a \emph{rank} on XBWT), the data structure of~\cite{optimal-rank-select} allows moving in $O(\log\log_w \sigma)$ time to the children of any node.  
See Ferragina et al.~\cite{XBWT} for a more detailed discussion of these operations.

After \emph{count}, the operation \emph{locate} can be solved given the ability of converting any XBWT position $i$ to the corresponding pre-order number $\preord u_i$ (but this is not the only option, see Section \ref{sec:locate}). 
To perform this conversion, we use a corrected version of Arroyuelo et al.'s strategy \cite[Sec. 5.1]{ANSalgor10}: we sample pre-order numbers and compute non-sampled values by visiting a small sub-tree using the XBWT primitives\footnote{Their solution has a problem that we fix here (personal communication with the authors): without using a tree decomposition, a single operation on the XBWT is not always sufficient to move to the next pre-order node.}. 
Fix a parameter $1\leq t \leq n$. 
%Broadly speaking, the idea is to explicitly store (on the corresponding XBWT positions) the pre-order of just $O(n/t)$ nodes. Non-sampled pre-orders can be retrieved by navigating the XBWT using parent and child operations described in the previous paragraph. 
We use the tree covering procedure described in~\cite[Sec. 2.1]{geary2006succinct} to decompose $\mathcal T$ in $\Theta(n/t)$ sub-trees containing $O(t)$ nodes each. Two sub-trees are either disjoint or intersect only at their common root. Let $\preord u_i$ be the root of a sub-tree, and consider the following quantities: (i) the pre-order identifier $\preord u_i$ and (ii) the number of nodes $size(\preord u_i)$ contained in the complete sub-tree (that is, down to the leaves of $\mathcal T$) rooted in $\preord u_i$. 
We store information (i) explicitly in XBWT order, for all sub-tree roots. We moreover store the partial sums of the values (ii) in XBWT order. All XBWT positions corresponding to a sub-tree root are moreover marked using a zero-order compressed bitvector supporting constant-time rank and select queries and taking $o(n) + O((n/t)\log n)$ bits of space~\cite{RRR}. This scheme supports retrieving in constant time the values (i) and (ii) associated with those XBWT positions. Overall, these structures take $o(n) + O((n/t)\log n)$ bits of space.
At this point, let $j$ be a XBWT position for which we want to compute the corresponding pre-order $\preord u_j$. Without loss of generality, we may assume that $j$ is not a sub-tree root (otherwise $\preord u_j$ is explicitly sampled). 
By iterating the XBWT parent operation, we first move upwards until finding the root $\preord u_k$ of the sub-tree $\mathcal T'$ (of size $O(t)$) containing $\preord u_j$. 
From here, using the XBWT navigation primitives we perform an Euler tour of $\mathcal T'$. 
We now show that the sampled values (i) and the partial sum on values (ii) are sufficient to reconstruct the pre-order value of all visited nodes along the Euler tour (in particular, $\preord u_j$).
We maintain a counter $PRE$, initialized at $PRE=\preord u_k$ at the beginning of the tour. At each step, counter $PRE$ will coincide with the pre-order number of the nodes seen for the first time along the visit. 
Assume we are on node $\preord u_z$ during the visit, and that the visit requires us to move to the next non-visited child $\preord u_{z'}$ of $\preord u_z$. 
We say that a node of $\mathcal T$ is \emph{non-root} if it is not the root of a sub-tree (\emph{root} otherwise).
Let $\preord u_{r_1}, \dots, \preord u_{r_c}$ be the $c$ root children of $\preord u_{z}$, between $\preord u_{z'}$ and its previous non-root sibling, in lexicographic order. 
If there is no previous non-root sibling, then $\preord u_{r_1}, \dots, \preord u_{r_c}$ are all the root siblings of $\preord u_{z'}$ preceding it. 
Note that this sequence could be empty.
By construction, $size(\preord u_{r_1}), \dots, size(\preord u_{r_c})$ are adjacent in the partial sum array, so we use this array to increase $PRE = PRE + 1 + \sum_{i=1}^c size(\preord u_{r_i})$ in constant time. Then, $PRE$ is precisely the pre-order value $\preord u_{z'}$: a pre-order visit of the complete subtrees rooted in $\preord u_{r_1}, \dots, \preord u_{r_c}$ would have the same effect on $PRE$.
The other case to consider is when $\preord u_z$ has no more non-visited non-root children (including the case where it is a leaf or it has no non-root children at all).
Then, before moving to the parent of $\preord u_z$ we increase $PRE = PRE + 1 + \sum_{i=1}^c size(\preord u_{r_i})$, where  $\preord u_{r_1}, \dots, \preord u_{r_c}$ are the root children of $\preord u_z$ following its last visited non-root children (if any, otherwise they are all root children of $\preord u_z$).

Note that the Euler tour can be implemented using $O(1)$ working space. This yields our claim. \qed

\section{Proof of Theorem \ref{th:r<nHk}}

\paragraph{\textbf{Claim}.} 
	The number $r$ of XBWT runs is always at most $\mathcal H^{wc}_k + \sigma^{k+1}$ for any $k\geq 0$, where $\mathcal H^{wc}_k$ is the trie's $k$-th order worst-case entropy (Definition \ref{def:Hk}). 

\paragraph{\textbf{Proof}.} 

	Let $OUT_i = out(\preord u_i)$, and let $RLE_c(i,j)$ be the number of $c$-run breaks in the sequence $X = OUT_i, \dots, OUT_j$: $RLE_c(i,j)$ increases by one unit for every $i \leq t \leq j$ such that $c\in OUT_t$ and either $t=j$ or  $c\notin OUT_{t+1}$. We denote $RLE(i,j) = \sum_{c\in\Sigma} RLE_c(i,j)$. Note that $r_c = RLE_c(1,n)$ and $r = RLE(1,n)$. For any partition $[i_1,i_2],[i_2+1, i_3], \dots, [i_m+1,i_{m+1}]$ of the interval $[1,n]$ into $m$ sub-intervals, it is easy to see that 
	\begin{equation}\label{r<sum rc}
	r = RLE(1,n) \leq \sum_{j=1}^m RLE(i_{j},i_{j+1})
	\end{equation}
	since the right-hand side has the same run breaks as the left-hand side, plus one more run break for the last occurrence of each character in each sub-interval. We now consider the partition into sub-intervals induced by the contexts of length $k$: we put in the same interval $cover(\rho) = OUT_i, \dots, OUT_j$ the outgoing labels of all nodes $\preord u_t$ having the same context $\rho = \pi_k[\preord u_t]$ (note that, by definition of $<$, such nodes form a consecutive range). To prove our thesis, we are going to show that $RLE(i,j) \leq \mathcal H^{wc}(cover(\rho)) + \sigma$ for any such interval $[i,j]$ corresponding to context $\rho$. 
	Let $n' = j-i+1$ and $n'_c$ be the number of occurrences of $c\in\Sigma$ in the sequence of sets $cover(\rho) = OUT_i, \dots, OUT_j$. We first prove $RLE_c(i,j) \leq n'_c\log_2(n'/n'_c) + 1$ for any character $c$ such that $n'_c>0$ (note: if $n'_c=0$ then $c$ does not contribute to $RLE(i,j)$ nor to the worst-case entropy of the interval). 
	
	Build a binary sequence $S[1,n']$ such that $S[t] = 1$ if and only if $c\in OUT_{i+t-1}$. 
	Letting $r_x(S)$ be the number of equal-letter maximal runs of symbol $x\in\{0,1\}$ in $S$, 
	by definition we have $r_1(S) = RLE_c(i,j)$. Note that $r_1(S) \leq r_0(S) +1$.
	Note also that $r_1(S) \leq n'_c$ and $r_0(S) \leq n'-n'_c$. From these inequalities we obtain that $RLE_c(i,j) = r_1(S) \leq \min\{n'_c, n'-n'_c+1\} \leq \min\{n'_c, n'-n'_c\} + 1$ always holds. 
	
	The next step is to prove $\min\{n'_c, n'-n'_c\} \leq n'_c\log_2(n'/n'_c)$. 
	We are going to prove this analytically by extending the domain of $n'_c$ and $n'$ to the whole $\mathbb R^+$, with the constraint $1\leq n'_c \leq n'$.
	If $n'_c < n'/2$, the inequality reduces to $n'_c \leq n'_c\log_2(n'/n'_c)$ which is obviously true in the considered range. If $n'_c \geq n'/2$, the inequality reduces to $n'-n'_c \leq n'_c\log_2(n'/n'_c)$. Let us define $\epsilon = n'_c/n'$. The inequality further simplifies to $f(\epsilon) = \epsilon - \epsilon \log_2\epsilon - 1 \geq 0$ for $0.5\leq \epsilon \leq 1$. The derivative $f'(\epsilon) = 1-\log_2\epsilon - \log_2e$ goes to zero in $f'(2/e) = 0$, is positive for $\epsilon < 2/e$ and negative for $\epsilon > 2/e$. Since $0.5 \leq 2/e \leq 1$, we obtain our claim: first, $f(0.5) = 0$, then $f(\epsilon)$ is increasing until $\epsilon = 2/e$, and finally it decreases until reaching $f(1) = 0$.
	
	From the above, we obtain that $RLE_c(i,j) \leq 1+ n'_c\log_2(n'/n'_c)$ for any character $c$ such that $n'_c>0$\ \footnote{Note that summing over all $c\in\Sigma$, this inequality implies $RLE(i,j) \leq \sigma + \sum_{c\in\Sigma} n'_c\log_2(n'/n'_c)$. The latter summation essentially coincides with the definition of zero-order empirical entropy on strings. However, on tries the interpretation of this quantity is not clear as in the string domain. The following steps yield a bound based on the worst-case entropy of the trie, which has a more straightforward interpretation.}. Since $(n'/n'_c)^{n'_c} \leq {n'\choose n'_c}$, we obtain $RLE_c(i,j) \leq 1+\log_2 {n'\choose n'_c}$. Summing both sides for all $c\in\Sigma$, we obtain $RLE(i,j) \leq \sigma + \sum_{c\in\Sigma}\log_2 {n'\choose n'_c} = \sigma + \mathcal H^{wc}(cover(\rho))$. 
	On the trie's paths there are in total at most $\sigma^k$ different contexts $\rho \in \Sigma^{k}$. Summing both sides of the inequality for all possible (at most) $\sigma^k$ contexts $\rho$ and applying Definition \ref{def:Hk} and Inequality \ref{r<sum rc}, we obtain $r \leq \mathcal H_k^{wc} + \sigma^{k+1}$. \qed

\section{Relations with Wheeler Automata}\label{sec:WA}

The smallest Wheeler Deterministic Finite-state Automaton (WDFA)~\cite{SODA20,GAGIE201767} equivalent to $\mathcal T$ can also be considered as a compressed  representation of the trie. 
This is the smallest DFA equivalent to the trie for which the co-lexicographic axioms (i) and (ii) defined at the beginning of Section \ref{sec:XBWT} hold~\cite{SODA20,GAGIE201767}.
In this section we show that these combinatorial objects and the XBWT  are deeply related. 
We start by introducing two equivalence relations between nodes that will play a fundamental role throughout the paper.
We write
$\preord u \equiv^r \preord v$ if and only if
%$\preord u$ and $\preord v$ have the same set of outgoing labels: 
$out(\preord u) = out(\preord v)$.
Note that the following property holds: for $i<n$, we have $\preord u_i \not\equiv^r \preord u_{i+1}$ if and only if $i$ is a run break. 
The second equivalence relation is a refinement of $\equiv^r$ and captures a slightly stronger relation than isomorphism: we write $\preord u \equiv \preord v$ if and only if $\lambda(\preord u) = \lambda(\preord v)$ and $\preord u \approx \preord v$.  
Clearly, $\equiv$ is a refinement of $\equiv^r$: if $\preord u \equiv \preord v$, then $\preord u \equiv^r \preord v$. 
The \emph{convex closure} $\equiv_<$ of $\equiv$ with respect to the order $<$ is defined as follows:
$\preord u_i \equiv_< \preord u_j$  if and only if $\preord u_i \equiv \preord u_j\  \wedge\  \forall\ \preord u_k \left(\min\{i,j\} < k < \max\{i,j\} \Rightarrow  \preord u_k \equiv \preord u_i\right)$.
The convex closures $\equiv^r_<$ and $\approx_<$ of $\equiv^r$ and $\approx$ are defined analogously. Note that the equivalence classes of $\equiv^r_<$ correspond to the RL-XBWT blocks. Note also that (see Figure \ref{fig:example})
$\equiv_<$, $\equiv^r_<$, and $\approx_<$ are refinements of $\equiv$, $\equiv^r$ and $\approx$, respectively, and $\equiv_<$ is a refinement of $\approx_<$, which in turn is a refinement of $\equiv^r_<$.
%From the automata-theory point of view, the relation $\equiv_<$ is the convex closure of the Myhill-Nerode equivalence relation (which in this paper coincides with $\approx$), with the further requirement that equivalent states must be reached by the same label. This 
Relation $\equiv_<$ has been introduced for the first time (with the symbol $\equiv_w$) by Alanko et al.~\cite{SODA20}, who prove that the quotient automaton $\mathcal T/_{\equiv_{<}}$ is the minimum WDFA equivalent to $\mathcal T$~\cite[Thm 4.1]{SODA20}.
We  show that $r$ is a lower bound to the size (number of edges) $\omega$ of such automaton:

\begin{theorem}\label{th:r<w}
	Let $\omega$ be the number of edges of the minimum WDFA recognizing the same language of $\mathcal T$. Then, $r \leq \omega$.
\end{theorem}
\begin{proof}
	Let $\mathcal T = (V,E)$.
	Consider any equivalence class $[\preord u]_{\equiv^r_<} = \{\preord u_i,\preord u_{i+1}, \dots, \preord u_{j}\}$. By definition of $\equiv^r_<$, all nodes in this class have the same children labels. It follows that the only run break in this class can be $\preord u_j$. This shows that $r \leq \sigma \cdot |V/_{\equiv^r_<}|$, because $r$ can increase by at most $\sigma$ only between two adjacent $\equiv^r_<$-classes. We can say more: between $[\preord u]_{\equiv^r_<}$ and the class immediately succeeding it in the ordering of the nodes, $r$ can increase at most by the number of children of $\preord u$ (since $\preord u_j$ can be a $c$-run only if $c$ is the label of a child of $\preord u_j$). It follows that $r$ can be upper-bounded as follows:
	$$
	r \leq \sum_{U \in V/_{\equiv^r_<}} |out(max(U))|
	$$
	where $max(U)$ returns the largest $\preord u\in U$ (by the ordering $<$). Now, since $\equiv_<$ is a refinement of $\equiv^r_<$ we have that 
	$$
	\sum_{U \in V/_{\equiv^r_<}} |out(max(U))| \leq \sum_{U \in V/_{\equiv_<}} |out(max(U))| = \omega
	$$
	from which the thesis follows. \qed
\end{proof}

An intriguing consequence of Theorem \ref{th:r<w} is that one can reduce the problem of indexing any acyclic Wheeler automaton $\mathcal A$ to the problem of indexing (the run-length XBWT of) the equivalent tree within $O(r)$ words of space: the resulting index will not be larger than $\mathcal A$. At a higher level, it is interesting to note that our technique collapses isomorphic subtrees that are adjacent in co-lexicographic order (see also Section \ref{sec:tree attr}). This is similar to the tunneling technique described by Alanko et al. \cite{tunneling} for Wheeler graphs. We conjecture that there is a deep link between our technique and theirs. 

\section{Tree Attractors}\label{sec:tree attr}
Let $S\in \Sigma^n$.
A string attractor \cite{attractors} is a set $\Gamma \subseteq [1,n]$ of the string's positions such that any substring $S[i,j]$ has at least one occurrence $S[i',j'] = S[i,j]$ 
%crossing at least one element $k \in \Gamma$, i.e. 
such that
$\Gamma \cap [i',j'] \neq \emptyset$. String attractors generalize most known dictionary compressors (for example, the run-length BWT, Lempel-Ziv 77, and straight-line programs), in the sense that a compressed representation of size $\alpha$ can be turned into a string attractor of size $O(\alpha)$~\cite{attractors}. Conversely, most compressibility measures can be upper-bounded by $O(\gamma\cdot \mathrm{polylog}\ n)$, where $\gamma$ is the size of the smallest string attractor \cite{kempa2019resolution,attractors,LATIN18}. 

Since string attractors capture the repetitiveness of a string, it is natural to try to generalize them to trees. We now propose such a generalization and exhibit a tree attractor of size $r$.

\begin{definition}\label{def:tree attr}
Let $\mathcal T = (V,E)$. A tree attractor is a subset $\Gamma \subseteq E$ such that any subtree $\mathcal T(U)$, with $U\subseteq V$, has at least one isomorphic occurrence $\mathcal T(U') = (U', E')$ such that $\Gamma \cap E' \neq \emptyset$.
\end{definition}

Let $\mathcal T = (V,E)$.
We define  $\Gamma^r = \{ (\preord u_i,\preord v)\in E\ :\ \exists c\in\Sigma\ |\ i\ is\ a\ c\mathtt{-}run\ break\ and\ \lambda(\preord v) = c\}$. 
%Set $\Gamma^r$ contains all the edges forming a run break. 
In Figure \ref{fig:example1}, the edges of $\Gamma^r$ are colored in red. We now show that $\Gamma^r$ is a tree attractor.

\begin{theorem}\label{th:Gamma^r is tree attr}
$\Gamma^r$ is  a tree attractor of size $|\Gamma^r| = r$.
\end{theorem}
\begin{proof} 
	The fact that $|\Gamma^r| = r$ follows from the very definitions of $\Gamma^r$ and $r$. 
	Let $\mathcal T(U) = (U,E'')$, with $U\subseteq V$, be a subtree of $\mathcal T$. If $E'' \cap \Gamma^r \neq \emptyset$ then we obtain our claim.
	Similarly, if the root of $\mathcal T(U)$ is $\preord u_n$ (the last node in the co-lexicographic order of the nodes of $\mathcal T$) then $n$ is a run break and all edges leaving $\preord u_n$ are in $\Gamma^r$. It follows that $E'' \cap \Gamma^r \neq \emptyset$ holds and we are done. 
	
	Let us therefore assume that $E'' \cap \Gamma^r = \emptyset$ and that the root of $\mathcal T(U)$ is $\preord u_i$, with $i<n$. Since no edge from $E''$ leaving $\preord u_i$ belongs to $\Gamma^r$, we have that $c\in out(\preord u_i) \Rightarrow c \in out(\preord u_{i+1})$. But then, since 
	$\preord u_i <_{pred} \preord u_{i+1}$ by Lemma \ref{lem:prec} it must 
	be the case that $child_c(\preord u_i) <_{pred} child_c(\preord u_{i+1})$ for all $c = \lambda(\preord v)$, where $(\preord u_i,v)\in E''$: 
	the children of $\preord u_i$ and $\preord u_{i+1}$ reached by following label $c$ must be adjacent in the co-lexicographic order of the tree. It is clear that we can repeat the above 
	reasoning to each such node $\preord v = child_c(\preord u_i)$ since, by assumption, no edge from $E''$ leaving $\preord v$ belongs to $\Gamma^r$. This procedure can be repeated until we visit the whole $\mathcal T(U)$. As a consequence, we obtain that $\mathcal T(U)$ has an isomorphic occurrence $\mathcal T(U') = (U', E')$ with root $\preord u_{i+1}$ in $\mathcal T$. If $E' \cap \Gamma^r \neq \emptyset$, we are done. Otherwise, we can repeat the whole reasoning to $\mathcal T(U')$, finding another isomorphic occurrence (rooted in $\preord u_{i+2}$). Note that the roots of this sequence of isomorphic trees are $\preord u_i <_{pred} \preord u_{i+1} <_{pred} \preord u_{i+2}, \dots$. By the finiteness of $\mathcal T$ and by the totality of $<$, this sequence cannot be infinite, therefore at some point we must stop finding a subtree $(\bar U, \bar E) = \mathcal T(\bar U) \approx \mathcal T(U)$ such that $\bar E \cap \Gamma^r \neq \emptyset$. \qed 
\end{proof}

Let $\gamma$ be the size of the smallest tree attractor and $\omega$ be the number of edges of the smallest Wheeler DFA equivalent to $\mathcal T$. By Theorems \ref{th:r<w} and \ref{th:Gamma^r is tree attr} we obtain  $\gamma  \leq r \leq \omega$.

\section{Proof of Lemma \ref{lem:cr}}

\paragraph{\textbf{Claim}.} 
There is a data structure taking $O(r\log n) + o(n)$ bits of space and supporting operation $cr(\wheel u,c)$ in $O(\log\sigma)$ time. 

\paragraph{\textbf{Proof}.} 
Consider our RL-XBWT representation of Definition \ref{def:RL-BWT}: $(ADD_q,DEL_q,\ell_q)_{q=1,\dots, \nblocks}$.
We mark in an entropy-compressed bitvector supporting constant-time rank and select queries~\cite{RRR} all nodes (in co-lexicographic order) that are the first in their XBWT block. 
Since the total number of XBWT blocks is $r' \leq 3r$ (Lemma \ref{lem:RL-BWT}), the bitvector takes $O(r\log n) + o(n)$ bits of space~\cite{RRR}. 

Let $\Sigma = \{c_1, \dots, c_\sigma\}$ be the original alphabet. We define a new alphabet $\Sigma' = \{/\} \cup \{c^-,c^+\ :\ c\in\Sigma\}$. Characters of $\Sigma'$ are sorted as follows: $c_1^- \prec  c_2^- \prec  \dots \prec  c_\sigma^- \prec  c_1^+ \prec  c_2^+ \prec  \dots \prec  c_\sigma^+ \prec  /$, i.e. all characters of the form $c^-$ come before those of the form $c^+$ and $/$ is the largest character. 
We build a sequence $S'$ over $\Sigma'$ by concatenating the characters of all sets $ADD_i$ and $DEL_i$ of our RL-XBWT, separating each block with a special symbol '/' as follows  ($\bigodot$ is the concatenation operator between strings and the concatenation order from each set is lexicographic):
$$
S' = \bigodot_{i=1}^{r'} \left( \left( \bigodot_{c\in ADD_i} c^+ \right) \bigodot \left( \bigodot_{c\in DEL_i} c^- \right) \bigodot / \right)
$$
Figure \ref{fig:S'} shows a running example.

\begin{figure}[h!]
	\centering
	\begin{tabular}{rcc}
		$S'$ & = & $a^+b^+c^+/a^-c^-/b^-/a^+c^+/a^-c^-/b^+c^+/b^-c^-/a^+/$
	\end{tabular}\caption{Sequence $S'$ obtained from the example of Figures \ref{fig:example1} and \ref{fig:example}.}\label{fig:S'}
\end{figure}

Clearly, $S'$ has $O(r)$ characters over an alphabet of size $\sigma' \in O(\sigma)$. We build over $S'$ 
a wavelet tree~\cite{WT}, taking $O(|S'|\log\sigma') \subseteq O(r\log n)$ bits of space and supporting rank and select operations in $O(\log\sigma)$ time.

%Let moreover $\bigodot$ be the concatenation operator between strings, and let $/\notin \Sigma$ be a new character lexicographically \emph{larger} than all other characters in $\Sigma$.
%We define the following two sequences on $\Sigma \cup \{/\}$, obtained by simply concatenating all characters in the sets $ADD_q$ and $DEL_q$ and separating different sets with the symbol '$/$':
%$$
%S_{ADD} = \bigodot_{i=1}^{\nblocks}\ \left( \left(\bigodot_{c\in ADD_i} c\right) \bigodot / \right)
%$$
%and
%$$
%S_{DEL} = \bigodot_{i=1}^{\nblocks}\ \left( \left(\bigodot_{c\in DEL_i} c\right) \bigodot / \right)
%$$
%where the operator $\bigodot$ concatenates characters from a set in any order (for example, lexicographic).
%By Lemma \ref{lem:RL-BWT}, the length of $S_{ADD}$ and $S_{ADD}$ is upper-bounded by $O(r)$. See Figure \ref{fig:SADD-SDEL} for a running example. 

%\begin{figure}[h!]
%	\centering
%	\begin{tabular}{rcc}
%		$S_{ADD}$ & = & abc///ac//bc//a/\\
%		$S_{DEL}$ & = & /ac/b//ac//bc//
%	\end{tabular}\caption{Sequences $S_{ADD}$ and $S_{DEL}$ obtained from the example of Figure \ref{fig:example}.}\label{fig:SADD-SDEL}
%\end{figure}

First, note that by definition all nodes within the same RL-XBWT block have the same answers to operation $cr(\wheel u,c)$. 
%We mark all nodes at the end of a block (in Wheeler order) in an entropy-compressed bitvector supporting constant-time predecessor queries~\cite{RRR}.
%Since the total number of blocks is $r' \leq 3r$ (Lemma \ref{lem:RL-BWT}), this bitvector takes $O(r\log n) + o(n)$ bits of space. 
With a constant-time predecessor on the bitvector marking the first nodes in each XBWT block,
we can therefore reduce $cr(\wheel u,c)$ to the analogous operation $cr'(i,c)$ on blocks, where this time $1 \leq i \leq \nblocks$ is the index of the RL-XBWT block the node $\wheel u$ belongs to and $cr'(i,c)$ is the answer to $cr(\wheel v,c)$ for any node $\wheel v$ in the $i$-th block. Now, let $i$ be a block number, and $j = S'.select_{/}(i)$ be the position in $S'$ containing the $i$-th occurrence of $/$. 
Note that in sequence $S'$ two consecutive occurrences of $c^+$ must be interleaved by exactly one occurrence of $c^-$ (in any position between those two $c^+$s).
Then, it is easy to see that the following holds: 

\begin{lemma}\label{lem:cr1}
	$S'.rank_{c^+}(j) - S'.rank_{c^-}(j)$ is equal to the number of edges labeled $c$ exiting any node in the $i$-th block (in particular, it is always either 0 or 1).
\end{lemma}

Let $S'.rank_{a,b}(i) = \sum_{a \leq d\leq b} S'.rank_{d}(i)$ be the number of characters belonging to the lexicographic range $[a,b]$ in $S'[1,i]$. 
This query is also known as \emph{three-sided range counting}.
A direct consequence of Lemma \ref{lem:cr1} is the following: 

\begin{corollary}\label{cor:rank<c}
	$S'.rank_{c_1^+,c^+}(j) - S'.rank_{c_1^-,c^-}(j)$ is equal to the number of edges labeled with all characters smaller than or equal to $c\in\Sigma$ exiting any node in the $i$-th block.
\end{corollary}

Wavelet trees support also query $S'.rank_{a,b}(i)$ in $O(\log \sigma)$ time~\cite{WT}. Corollary \ref{cor:rank<c} solves precisely query $cr'(i,c)$, so we obtain our claim. \qed

\vspace{10pt}

Even if we will not need it in our index, we note that binary search on Corollary \ref{cor:rank<c} can be used to solve also the following operation in $O(\log^2\sigma)$ time: 
\texttt{child label} $cl(\wheel u,k)$, which returns the label of the edge connecting $\wheel u$ with its $k$-th (in lexicographic order) child. This operation could be useful, for example, to list the (labels of the) children of any $\wheel u$ within $O(r\log n) + o(n)$ bits of space.

\section{Proof of Lemma \ref{lemma:isd}}

\paragraph{\textbf{Claim}.} 
	The tree representation~\cite{suctrees} supports also operation $ISD(\preord u, \preord v,\preord u')$ in $O(1)$ time at no additional space usage. 

\paragraph{\textbf{Proof}.} 
	The tree representation~\cite{suctrees} stores the Balanced Parentheses Sequence (BPS) representation of the tree topology, augmented with additional (light) structures.
	Let $i_{\preord u}$, $i_{\preord v}$, $i_{\preord u'}$ be the positions of the open parentheses corresponding to nodes $\preord u,\preord v$, and $\preord u'$ in the BPS representation of the tree topology.
	Since $\preord u\approx \preord u'$ the parentheses substring representing $\preord u$ and its descendants is equal to the one representing $\preord u'$ and its descendants. Then, it must be the case that $i_{\preord v} - i_{\preord u} = i_{\preord v'} - i_{\preord u'}$, therefore $i_{\preord v'} = i_{\preord v} - i_{\preord u} + i_{\preord u'}$. 	The representation~\cite{suctrees} allows moving between positions in the BPS sequence and pre-order ranks in constant time, so our thesis follows. \qed

\section{Proof of Lemma \ref{lemma:isc}}

\paragraph{\textbf{Claim}.} 
	Operation $ISC(\preord u_i, k)$ can be supported in $O(1)$ time and $O(r\log n) + o(n)$ bits of space.

\paragraph{\textbf{Proof}.} 
	For brevity, let $OUT_k = out(\preord u_k)$.
	For each node $\preord u_i$ such that $i<n$ is a run break (i.e. $out(\preord u_i) \neq out(\preord u_{i+1})$), we build the following two bitvectors: 
	$$
	S_i^1 = \bigodot_{c\in OUT_{i}} c \stackrel{?}{\in} OUT_{i+1} 
	$$ 
	and 
	$$
	S_i^2 = \bigodot_{c\in OUT_{i+1}} c \stackrel{?}{\in} OUT_{i} 
	$$ 
	where the operator $\bigodot$ visits characters in lexicographic order and where $c \stackrel{?}{\in} A$ equals the symbol '1' if $c\in A$ and '0' otherwise. In other words, $S_i^1$ marks an outgoing label of  $\preord u_i$ with a bit 1 if it is also an outgoing label of $\preord u_{i+1}$ and with a bit 0 otherwise (similar for $S_i^2$).
	We concatenate these two bit-sequences and further concatenate all such $S_i^1S_i^2$ in pre-order (that is, according to the pre-order number $\preord u_i$, rather than on $i$) in a single sequence $S$ of length $|S| \leq n$. 
	We furthermore use a bitvector $B_1$ of length $n$ to mark \emph{in pre-order} the nodes that are run-breaks (i.e. nodes $\preord u_i$ for which we built $S_i^1S_i^2$), and a bitvector $B_2$ of length $|B_2| = |S| \leq n$ to mark the boundaries of each $S_i^1$ and $S_i^2$ inside sequence $S$. We build on the two bitvectors the entropy-compressed representation of Raman et al.~\cite{RRR}, which answers \emph{rank} and \emph{select} queries in constant time. Since those bitvectors have length at most $n$ and have $O(r)$ bits set, the structure~\cite{RRR} uses $o(n) + O(r\log n)$ bits~\cite{RRR}. Using $S$, $B_1$, and $B_2$, we can retrieve in constant time 
	%(using rank and select operations on the two bitvectors) 
	the (boundaries in $S$ of the) two sequences $S_i^1$ and $S_i^2$ associated with any pre-order node $\preord u_i$ that is a run-break.
	We use Raman et al.'s representation~\cite{RRR} to represent also sequence $S$. Note that $S$ has one bit equal to 0 for each $c\in OUT_i$ such that $c\notin OUT_{i+1}$ and for each $c\in OUT_{i+1}$ such that $c\notin OUT_{i}$. It follows that $S$ has at most $O(r)$ bits equal to 0, therefore the entropy-compressed data structure~\cite{RRR} uses $o(n) + O(r\log n)$ bits to represent it. 
	
	We now show how to answer $ISC(\preord u_i, k)$.
	Let $c$ be the $k$-th (in lexicographic order) outgoing label of $\preord u_i$.
	We first retrieve in constant time 
	(the boundaries in $S$ of) $S_i^1$ and $S_i^2$. 
	Note that we can assume $S_i^1[k] = 1$ since, by assumption in our query definition, $c$ is an outgoing label of $\preord u_{i_+1}$. Let $S_i^1[k]$ be the $j$-th bit equal to '1' in $S_i^1$ (we can find $j$ in constant time with a \emph{rank} query). Then, it must be the case that the $j$-th bit equal to '1' $S_i^2[t]$ is such that the $t$-th outgoing label of $\preord u_{i+1}$ is equal to $c$ (note: by the way we constructed those two sequences, the corresponding bits set in $S_i^1$ and $S_i^2$ correspond to the same labels). We can find $t$ in constant time with a select operation on $S_i^2$. 
	Finally, we return $t$.
	\qed

\section{Proof of Lemma \ref{lem:toehold}}

\paragraph{\textbf{Claim}.} 
There is a data structure taking $O(r\log n) + o(n)$ bits of space on top of the succinct tree topology of Navarro and Sadakane~\cite{suctrees} that, 
given a pattern $P\in \Sigma^m$, returns the co-lexicographic range $[\wheel \ell,\wheel r]$ of nodes reached by a path labeled $P$, as well as $\preord u_{\wheel\ell}$, in $O\left(m\log\sigma\right)$ time.

\paragraph{\textbf{Proof}.} 
	Finding the range of nodes $[\wheel \ell, \wheel r]$ reached by a pattern requires, as building block, being able to count the number of occurrences of a character $c$ in a prefix $out(\preord u_1), \dots, out(\preord  u_i)$ of the XBWT, an operation we denote as $rank_c(i)$~\cite{XBWT}. 
	Moreover, in order to find node $\preord u_{\wheel \ell}$ we will find, given an index $i$ and a character $c\in\Sigma$, the minimum $i'\geq i$ such that $c\in out(\preord u_{i'})$. 
	Such $i'$ will always exist in our application below. 
	We denote this operation as $successor_c(i) = i'$. We now show how to solve these operations. 
	
	\paragraph{\textbf{rank}.}
	We first show how to support $rank_c(i)$ in  $O(r\log n) + o(n)$ bits of space and $O(\log\sigma)$ time (it is actually possible to improve upon this running time, but for us  $O(\log\sigma)$ will be sufficient due to the complexity of operation $cr(\wheel u,c)$, Lemma \ref{lem:cr}).
	
	We use the same sequence $S'$ defined in the proof of Lemma \ref{lem:cr}, represented with a wavelet tree, as well as the entropy-compressed bitvector marking nodes (in co-lexicographic order) that are the first in their XBWT block.
	
	Consider any occurrence of a character $c^+$ in $S'$, with $c\in\Sigma$, belonging to XBWT block $j$ (that is, between the $(j-1)$-th and $j$-th occurrence of $/$), and let $\preord u_{j'}$ be the first node in the $j$-th XBWT block. We explicitly store $rank_c(j'-1)$ in correspondence to this occurrence of $c^+$ (if $j'=1$, then we take $rank_c(j'-1) = 0$). Storing all these partial ranks takes $O(r)$ words of space in total. 
	
	Now, it is not hard to see that all these structures allow us to compute $rank_c(i')$ in $O(\log \sigma)$ time for any $c\in\Sigma$ and $1\leq i' \leq n$. First, we find the XBWT block $i$ containing node $\preord u_{i'}$ (constant time on the bitvector marking the first node of each block). Then, we find in $S'$ the occurrences of $c^-$ and of $c^+$ that immediately precede the $i$-th symbol '/' ($O(\log\sigma)$ time using rank and select operations). 
	If there are no such occurrences of $c^+$, then $rank_c(i')=0$.	
	We consider two other cases. 
	
	(A) The occurrence found of $c^+$ is to the right of that of $c^-$, or there are no such occurrences of $c^-$. 
	Let $j^+$ be the the XBWT block containing such occurrence of $c^+$.
	This means that all nodes contained in the XBWT blocks from the $j^+$-th to the $i$-th (included) have an outgoing edge labeled $c$.
	Let $\preord u_{j'}$ be the first node of the $j^+$-th XBWT block (found in constant time using our bitvector).
	Then, $rank_c(j'-1)$ is explicitly stored and we obtain $rank_c(i') = rank_c(j'-1) + (i'-j') +1$. 
	
	(B) The other case to be considered is the one where the occurrence found of $c^-$ is to the right of that of $c^+$. Let $j^-$ be the XBWT block containing such occurrence of $c^-$. Then, all XBWT blocks from the $j^-$-th to the $i$-th (included) do not have an outgoing edge labeled with $c$.
	Note that $j>1$, since there is an occurrence of $c^+$ before the $j^-$-th block. 
	Then, nodes in the $(j^--1)$-th block do have an outgoing edge labeled $c$. 
	Let $\preord u_{j'}$ be the first node in the $j^-$-th block (found in constant time using our bitvector). Then, $rank_c(i') = rank_c(j'-1)$, which reduces to case (A).
	
	\paragraph{\textbf{successor}.} We  show how to solve $successor_c(i')$. If $rank_c(i') = rank_c(i'-1)+1$ (where $rank_c(0)=0$), then 
	node $\preord u_{i'}$ has an outgoing edge labeled $c$ and
	we return $i'$.
	Otherwise, we find the next XBWT block containing nodes that have an outgoing edge labeled $c$. Let $i$ be the XBWT block containing node $\preord u_{i'}$ (found in constant time using our bitvector). Using one rank and one select operation on $S'$, we find the occurrence of $c^+$ that immediately follows the $i$-th occurrence of $/$ in $S'$ (in the application below, such an occurrence of $c^+$ will always exist). Let $j$ be the XBWT block containing this occurrence of $c^+$ (found with a rank operation on $S'$ to count the number of '/' preceding the occurrence of $c^+$ and adding 1 to the result). Let moreover $\preord u_{j'}$ be the first node in the $j$-th XBWT block (found in constant time using our bitvector). We return $j'$.

	\paragraph{\textbf{computing $[\wheel \ell, \wheel r]$}.} 
	Recall that $\wheel \ell$ is the co-lexicographic rank (in the list of sorted nodes) of the first pattern occurrence. Similarly, $\wheel r$ is the rank of the last such node. 
	We show how to find the co-lexicographic range $[\wheel \ell, \wheel r]$ of all nodes reached by a given pattern. 
	The algorithm (known as \emph{backward search}) is based on the observation that labels occur in the same order in the XBWT and in the sequence $\lambda(\preord u_1), \dots, \lambda(\preord u_n)$~\cite{XBWT,GAGIE201767}, see Figure 1. 
	Moreover, the nodes reached by a path labeled $P \in \Sigma^*$ always form a consecutive range with respect to the co-lexicographic order~\cite{GAGIE201767}. These observations lead to the following algorithm, first described in~\cite{XBWT} (on trees). 
	First, note that characters in $\lambda(\preord u_1), \dots, \lambda(\preord u_n)$ are sorted (i.e. clustered in increasing order). We store in an array $C$ a total of $\sigma \leq r$ integers recording the starting point of every distinct character in this sequence. 
	At this point, given the co-lexicographic range $[\wheel \ell,\wheel r]$ of nodes reached by a path labeled $P\in\Sigma^*$, to extend it with character $c\in\Sigma$ we map the characters equal to $c$ contained in $out(\preord u_{\wheel \ell}), \dots, out(\preord u_{\wheel r})$ to the corresponding range $\lambda(\preord u_{\wheel \ell'}), \dots, \lambda(\preord u_{\wheel r'})$ using just two \emph{rank} queries and one access to array $C$. The result $[\wheel \ell', \wheel r']$ is the range of nodes reached by a path labeled $P\cdot c$. At the beginning, the algorithm starts with $P=\epsilon$ (empty pattern) and $[\wheel \ell,\wheel r] = [1,n]$.
	Crucially, note that this procedure returns only the range of \emph{ranks} (in co-lexicographic order)  $[\wheel \ell,\wheel r]$ of the nodes reached by a path labeled $P$. To obtain their pre-order identifiers $\preord u_{\wheel \ell}, \dots, \preord u_{\wheel r}$ we will need the more complex \emph{locate} queries, discussed in Section \ref{sec:locate}.
	
	\paragraph{\textbf{computing $\preord u_{\wheel\ell}$}.}
	We  show how to extend the above procedure in order to also compute $\preord u_{\wheel\ell}$.
	At the beginning, we start with an empty pattern $P=\epsilon$ and its range $[1,n]$. Then, $\preord u_1 = 1$ is the root. 
	Assume now that we have computed the range $[\wheel \ell, \wheel r]$ of a pattern $P$, and that we know the pre-order node $\preord u_{\wheel \ell}$. 
	We extend $P$ with letter $c$ and obtain the range  $[\wheel \ell', \wheel r']$ of $P\cdot c$ with an extension step described above (assume that the range is not empty, otherwise the search stops). Then, we find in $O(\log\sigma)$ time with a successor query (read above) the smallest $i$ in the range  $[\wheel \ell, \wheel r]$ such that $c\in out(\preord u_i)$.	 
	Note that such a successor always exists, since we assume that $[\wheel \ell', \wheel r']$ is not empty.
	If $i=\wheel \ell$, then we simply descend to the corresponding child of $\preord u_{\wheel \ell}$ with $\preord u_{\wheel \ell'} = cbr(\preord u_{\wheel \ell},cr(\wheel \ell, c))$ in  $O\left(\log\sigma\right)$ time and $2n+o(n)+O(r\log n)$ bits of space (by Operations \ref{op:cbr} and \ref{op:cr}). Otherwise, $i>\wheel \ell$. But then, co-lex node $i$ is the first in a run of nodes having an outgoing edge labeled $c$ (that is, co-lex node $i-1$ does not have an outgoing edge labeled $c$). 
	We can therefore explicitly store all those pre-order nodes, since there are at most $O(r)$ of them, and retrieve $\preord u_i$ in constant time. Finally, we  descend to the edge labeled $c$ of $\preord u_{i}$ with $\preord u_{\wheel \ell'} = cbr(\preord u_{i},cr(i, c))$ in  $O\left(\log\sigma\right)$ time and $2n+o(n)+O(r\log n)$ bits of space (by Operations \ref{op:cbr} and \ref{op:cr}). \qed

\section{Proof of Lemma \ref{lem:find colored}}

\paragraph{\textbf{Claim}.} 
 There is a data structure taking $O(r\log n) + o(n)$ bits of space on top of the succinct tree topology of Navarro and Sadakane~\cite{suctrees} and answering the following queries in $O(1)$ time. Given a pre-order node $\preord u_i$ with $i<n$:
\begin{itemize}
	\item[(a)] 
	If $\preord u_i$ is not colored, find a colored node $\preord u_j \neq \preord u_i$ in the complete subtree rooted in $\preord u_i$ such that no node on the path from $\preord u_i$ to $\preord u_j$ is colored (except $\preord u_j$), or report that $\preord u_j$ does not exist.
	\item[(b)] Find the lowest ancestor $\preord u_j$ of $\preord u_i$ such that the complete subtree rooted in $\preord u_j$ contains a colored node. Note that such a node always exists, since the root is always blue. 
\end{itemize}

\paragraph{\textbf{Proof}.} 
	Consider the Balanced Parentheses Sequence (BPS) representation of the tree. To answer (a), it is sufficient to mark in a bitvector $B$ all open parentheses corresponding to a colored node (note: we mark $O(r)$ parentheses). By the definition of BPS, $I = \preord u_i$ corresponds to the $I$-th open parenthesis in the sequence. 
	If the $I$-th open parenthesis is marked, then we return $\preord u_i$. otherwise, let $J$ be the position of the marked open parenthesis immediately following the $I$-th. If the position of $J$ falls inside the BPS range of node $\preord u_i$ (that is, between its corresponding open and close parentheses), then $J = \preord u_j$ is the descendant of $\preord u_i$ that we are looking for. Otherwise, the complete subtree rooted in $\preord u_i$ does not contain colored nodes and we report that $\preord u_j$ does not exist. By using Raman et al.'s entropy-compressed representation~\cite{RRR}, bitvector $B$  takes $O(r\log n) + o(n)$ bits and answers successor queries in constant time. All operations on the BPS representation (in particular, finding matching pairs of open/close parentheses)  take constant time~\cite{suctrees}.
	
	We now show how to answer (b). 
	Consider again the $I$-th open parenthesis, i.e. $I = \preord u_i$. The idea is to find the $K$-th open parenthesis that immediately precedes the $I$-th and that is also marked. Let $\preord u_k = K$, and let $\preord u_t = LCA(\preord u_i, \preord u_k)$. Then, if another node $\preord u_{t'} \neq \preord u_t$ on the path $\preord u_t \rightsquigarrow \preord u_i$ is such that the complete subtree rooted in $\preord u_{t'}$ contains a colored node $\preord u_{k'}$, it must be the case that $\preord u_{k'}$ appears after $\preord u_{i}$ in pre-order (otherwise we would have found the rightmost such node in place of $\preord u_k$). To complete the procedure we must therefore also find the $K'$-th open parenthesis that immediately succeeds the closing parenthesis of $\preord u_i$ and that is also marked. Let $\preord u_{k'} = K'$, and let $\preord u_{t'} = LCA(\preord u_i, \preord u_{k'})$. The answer to our query is the deepest node between $\preord u_{t}$ and $\preord u_{t'}$ (this requires computing $depth(\preord u_{t})$ and $depth(\preord u_{t'})$). Note that all operations take constant time and that we use the same structures defined for query (a). 
	Again, all operations on the BPS representation (in particular: matching parentheses, LCA, $depth$)  take constant time~\cite{suctrees}.
	\qed

\section{Proof of Lemma \ref{lem:uj red}}\label{app:uj red}

\paragraph{\textbf{Claim}.} 
In Lemma \ref{lem:find colored} (a), if $\preord u_j$ exists then $\preord u_j$ must be red and not blue.

\paragraph{\textbf{Proof}.} 

Assume that the complete subtree rooted in $\preord u_i$ contains a colored node $\preord u_j \neq \preord u_i$ such that no node other than $\preord u_j$ on the path $\preord u_i \rightsquigarrow \preord u_j$ of length (number of nodes) $k\geq 2$ is colored. We are going to prove that $\preord u_j$ is red. By assumption in Lemma \ref{lem:find colored} (a), we have $i<n$. We prove the property inductively on the length $k$ of the path. Assume $k=2$ (that is, $\preord u_j$ is child of $\preord u_i$), and let $c = \lambda(\preord u_j)$.
Since by assumption $\preord u_i$ is not red and $i<n$, then $c\in out(\preord u_{i+1})$. By Lemma \ref{lem:prec} we have that $\preord u_{j+1} = child_c(\preord u_{i+1})$. But then, $c = \lambda(\preord u_{j+1}) = \lambda(\preord u_{j})$, therefore $\preord u_{j}$ cannot be blue. Since $\preord u_{j}$ is colored,  it must be the case that $\preord u_{j}$ is red (and not blue).

Let $k>2$, and let $\preord u_i \rightarrow \preord u_{i'} \rightsquigarrow \preord u_j$ be the path from $\preord u_i$ to $\preord u_j$, where $\preord u_{i'}$ is child of $\preord u_{i}$. 
Let $c = \lambda(\preord u_{i'})$.
Since $i<n$ and $\preord u_i$ is not red, we conclude that $c\in out(\preord u_{i+1})$. By Lemma \ref{lem:prec} we have that $\preord u_{i'+1} = child_c(\preord u_{i+1})$. Then, this implies that $i'<n$ therefore we can apply our inductive hypothesis to the path $\preord u_{i'} \rightsquigarrow \preord u_j$ of length $k-1$ and conclude that $\preord u_{j}$ is red and not blue. \qed

\section{Proof of Lemma \ref{lem:adj paths}}\label{app:cor:adj paths}

\paragraph{\textbf{Claim}.} 

Let $\Pi = \preord u_{i_1} \rightarrow \preord u_{i_2} \rightsquigarrow \preord u_{i_k}$, with $i_j<n$ for some $1\leq j \leq k$, be a path of length $k$ without blue nodes other than (possibly) $\preord u_{i_1}$ and without red nodes other than (possibly) $\preord u_{i_k}$. Then, $\preord u_{i_1+1} \rightarrow \preord u_{i_2+1} \rightsquigarrow \preord u_{i_k+1}$ is a path in the tree (adjacent to $\Pi$). 

\paragraph{\textbf{Proof}.} 

Let us break the path into two subpaths, overlapping by node $\preord u_{i_j}$: $\Pi' = \preord u_{i_1} \rightsquigarrow \preord u_{i_j}$ and $\Pi'' = \preord u_{i_j} \rightsquigarrow \preord u_{i_k}$. Intuitively, we break the proof for the two sub-paths since the proof for $\Pi'$ will use the absence of blue nodes (induction moves towards the root), while the proof for $\Pi''$ will use the absence of red nodes (induction moves towards the leaves).

Note that the following properties hold on the two individual subpaths: (1) in both $\Pi'$ and $\Pi''$, only the first node might be blue and only the last node might be red. (2) the last node $\preord u_{z'}$ of $\Pi'$ is such that $z'<n$, and the first node $\preord u_{z''}$ of $\Pi''$ is such that $z''<n$.
Note also that $i_j$ might coincide with $i_1$, $i_k$, or both. In this case, one of the two subpaths (or both) reduces to a single node. 
We prove the lemma separately for these two subpaths.

(Subpath $\Pi'$) We prove the property by induction on the number $t$ of nodes in the subpath. If $t=1$ the claim is immediate, since by assumption the only node $\preord u_{z'}$ in the subpath is such that $z'<n$, thus $\preord u_{z'+1}$ exists. 

Let therefore $\Pi' = \preord u_{j_1} \rightsquigarrow \preord u_{j_{t-1}} \rightarrow \preord u_{j_{t}}$ have length $t\geq 2$. By assumption, $j_{t}<n$ and no node other than (possibly) $\preord u_{j_{1}}$ is blue: it follows that $\lambda(\preord u_{j_{t}}) = \lambda(\preord u_{j_{t}+1})$. Consider the parents of these two nodes, $\pi(\preord u_{j_{t}}) = \preord u_{j_{t-1}}$ and $\pi(\preord u_{j_{t}+1}) = \preord v$. 
Since $\lambda(\preord u_{j_{t}}) = \lambda(\preord u_{j_{t}+1})$, 
by co-lexicographic Axiom (ii) it must be the case that $\preord u_{j_{t-1}} < \preord v$. We can say more: since by assumption $\preord u_{j_{t-1}}$ is not red, it must be the case that $\preord u_{j_{t-1}} <_{pred} \preord v$, i.e. that $\preord v = \preord u_{j_{t-1}+1}$. Assume, for contradiction, that there exists a node $\preord w$ such that $\preord u_{j_{t-1}} <_{pred} \preord w < \preord v$. Let $c = \lambda(\preord u_{j_{t}}) = \lambda(\preord u_{j_{t}+1})$.  We have two cases. If $c\in out(\preord w)$, then by co-lexicographic Axiom (ii) it must be the case that $\preord u_{j_{t}} <_{pred} child_c(\preord w) < \preord u_{j_{t}+1}$, a contradiction. If $c\notin out(\preord w)$, then $out(\preord u_{j_{t-1}}) \neq out(\preord w)$, therefore $\preord u_{j_{t-1}}$ is red: also a contradiction. We obtained that $\preord u_{j_{t-1}+1} \rightarrow \preord u_{j_{t}+1}$ is an edge in the tree and, in particular, $j_{t-1}<n$. We can therefore apply the inductive hypothesis to the subpath $\preord u_{j_1} \rightsquigarrow \preord u_{j_{t-1}}$ of length $t-1$ and obtain that $\preord u_{j_1+1} \rightsquigarrow \preord u_{j_{t-1}+1}$ is a path in the tree. Merging these two results, we obtain that $\preord u_{j_1+1} \rightsquigarrow \preord u_{j_{t-1}+1} \rightarrow \preord u_{j_{t}+1}$ is a path in the tree (adjacent to $\Pi'$).

(Subpath $\Pi''$) We prove the property by induction on the number $t$ of nodes in the subpath. If $t=1$ the claim is immediate, since by assumption the only node $\preord u_{z'}$ in the subpath is such that $z'<n$, thus $\preord u_{z'+1}$ exists. 

Let therefore $\Pi'' = \preord u_{j_1} \rightarrow \preord u_{j_2} \rightsquigarrow \preord u_{j_{t}}$ have length $t\geq 2$. By assumption, $j_{1}<n$ and no node other than (possibly) $\preord u_{j_t}$ is red.

Let $c = \lambda(\preord u_{j_2})$.
Since by assumption $\preord u_{j_1}$ is not red and $j_1<n$, then $c\in out(\preord u_{j_1+1})$. Then, by Lemma \ref{lem:prec} we have that $\preord u_{j_2+1} = child_c(\preord u_{j_1+1})$, thus $\preord u_{j_1+1} \rightarrow \preord u_{j_2+1}$ is an edge in the tree. In particular, $j_2<n$. By our inductive hypothesis, $\preord u_{j_2+1} \rightsquigarrow \preord u_{j_t+1}$ is a path in the tree. Merging these two results,  we obtain that $\preord u_{j_1+1} \rightarrow \preord u_{j_2+1} \rightsquigarrow \preord u_{j_t+1}$ is a path in the tree (adjacent to $\Pi''$). 

To conclude, we merge the two results obtained for $\Pi'$ and $\Pi''$ and obtain our claim: $\preord u_{i_1+1} \rightarrow \preord u_{i_2+1} \rightsquigarrow \preord u_{i_k+1}$ is a path in the tree (adjacent to $\Pi$).\qed

\section{Examples of Climb, Section \ref{sec:locate}}\label{app:examples climb}

\paragraph{Example of Case 1} Consider Figure \ref{fig:example1}, and suppose we want to compute $\phi(2)$. First, we find a (any) red descendant of 2: let's say we pick node 14 (the same reasoning holds with red node 3). Note that we have explicitly stored (orange dashed arrow) $\phi(14) = 6$. Note moreover that the path connecting 2 and 14 has length 1 and is labeled with string $S = b$. Lemma \ref{lem:adj paths} tells us that, along the path labeled $S$ connecting $\phi(2)=3$ and $\phi(14)=6$, the nodes are always adjacent in co-lexicographic order with the relative nodes in the path $2 \rightsquigarrow 14$. By applying our formula, we obtain 
$$
\phi(2) = LAQ(\phi(14), depth(14) - depth(2)) = LAQ(6,1) = 3
$$

\paragraph{Example of Case 2.1} Consider Figure \ref{fig:example1}, and suppose we want to compute $\phi(\preord u_i) = \phi(24)$. Node $\preord u_j = 1$ is the lowest ancestor of 24 such that the complete subtree rooted in 1 contains colored nodes. In this particular case, 1 is blue so we follow the explicit edge $\phi(1) = 2 = \preord u_{j+1}$. 
%Since no node on the path $\Pi = 1 \rightarrow 22 \rightsquigarrow 24$ is red, this path is adjacent in Wheeler order with the isomorphic path  $2 \rightarrow 14 \rightsquigarrow 16$. 
We moreover find the successor $\preord u_k$ of 1 in $\Pi = 1 \rightarrow 22 \rightsquigarrow 24$ with $t = depth(24) - depth(1) = 3$ and $\preord u_k = LAQ(24, t-1) = 22$. 
Since 1 is not red and 22 is the second child of 1, nodes 1 and 2 have the same outgoing labels and therefore the node 14 on the path $2 \rightsquigarrow 16 = \preord u_{i+1}$ must be the second child of 2.
By definition of $\preord u_j = 1$, no node in the complete subtree rooted in 22 is colored: this subtree is therefore isomorphic with the complete subtree rooted in 14. It follows that the relative position of $\preord u_{i+1} = 16$ in the subtree rooted 14 is the same as that of $\preord u_{i} = 24$ in the subtree rooted 22: we can therefore find node 16 with an isomorphic descendant query.

\paragraph{Example of Case 2.2.1} 
Consider Figure \ref{fig:example1}, and suppose we want to compute $\phi(\preord u_i) = \phi(5)$. 
Node $\preord u_j = 3$ is the lowest ancestor of 5 such that the complete subtree rooted in 3 contains a colored node.
Let $t = depth(\preord u_i) - depth(\preord u_j) = 2$. We find $\preord u_k = LAQ(\preord u_i,t-1) = 4$. 
Node $\preord u_k = 4$ is a $\phi$-sample of type 2. Then, $\phi(\preord u_k) = \preord u_{k+1} = 11$ is stored explicitly and we retrieve it in constant time. 
By definition of $\preord u_j$, no node in the subtree rooted in 4 is colored. Then, this subtree and the one rooted in 11 are isomorphic and we can find $\phi(\preord u_{i+1}) = 12$ with an isomorphic descendant query.

\paragraph{Example of Case 2.2.2} 
Consider Figure \ref{fig:example1}, and suppose we want to compute $\phi(\preord u_i) = \phi(6)$. 
Node $\preord u_j = 3$ is the lowest ancestor of 6 such that the complete subtree rooted in 3 contains a colored node. In this particular case, $\preord u_k$ coincides with $\preord u_i$, and $\phi(\preord u_k)$ is not a $\phi$-sample of type 2. In fact, as proved above, $\phi(\preord u_j) = \preord u_{j+1} = 4$ (which we retrieve in constant time, being it a $\phi$-sample of type 1) is adjacent in co-lexicographic order to node 3. Now, $\preord u_k = 6$ and $\preord u_{k+1} = 5$ are both reached by following label $b$ from $\preord u_j = 3$ and $\preord u_{j+1} = 4$, respectively. Since 3 is red and $3 <_{pred} 4$, we can find $\preord u_{k+1} = 5$ with an isomorphic child operation. Finally, as noted in the previous examples the subtrees rooted in 5 and 6 are isomorphic, so we can find $\phi(\preord u_i) = \phi(6) = 5$ in constant time with an isomorphic descendant query.

\section{Entropy bound}\label{app:entropy bound}

We show how to bound the size of our index as a function of $\mathcal H^{wc}_k$.

\begin{lemma}\label{cor:r<nHk}
For any $0 < \alpha < 1$ and $0 \leq k \leq \max\{0,\alpha\log_\sigma n-1\}$ it holds $r\leq 2\mathcal H^{wc}_k + o(n/\log^c n)$ for any constant $c>0$. 
\end{lemma}
\begin{proof}
 
If $0 < \alpha < 1$ and $0\leq k \leq \alpha\log_\sigma n-1$, then the corollary follows immediately from Theorem \ref{th:r<nHk}: $r \leq \mathcal H^{wc}_k + \sigma^{\alpha\log_\sigma n} = \mathcal H^{wc}_k + n^{\alpha} \leq \mathcal H^{wc}_k + o(n/\log^c n)$ for any constant $c>0$.

However, for large $\sigma$ the interval $[0,\alpha\log_\sigma n-1]$ could be empty. To prove the claim, we have to give a useful bound in the case $k=0$.
In this case, Theorem \ref{th:r<nHk} yields $r \leq \mathcal H^{wc}_0 + \sigma$. The problem is that $\sigma$ could be $\Theta(n)$; the solution is to note that, in this case, also $\mathcal H^{wc}_0$ must be large. 
In the following we prove that the bound $r\leq 2\mathcal H^{wc}_0 + 1$ holds. This will prove the claim.

We can assume the number of nodes to be $n\geq 2$, otherwise the tree is either empty or composed of the root only and both $r$ and $\mathcal H^{wc}_0$ are equal to 0.
We can moreover assume $\sigma \geq 2$, since character $\#$ does not label any edge and there are at least $2$ nodes. 

Let $n_c$ be the number of edges labeled $c$.
Note that $n_{\#}=0$ since $\#$ does not label any edge.
By definition, $\mathcal H^{wc}_0 = \sum_{c\in\Sigma} \log_2 {{n-1}\choose{n_c}} = \sum_{c\in\Sigma-\{\#\}} \log_2 {{n-1}\choose{n_c}}$.

If $\sigma = 2$, then the tree is a unary path, $r=1$, and $\mathcal H^{wc}_0 = 0$. The claim follows. 
We can therefore assume $\sigma \geq 3$.
Since we assume the alphabet to be \emph{effective} we have then $n-1\geq 2$ and $1 \leq n_c < n-1$, therefore ${{n-1}\choose{n_c}} \geq 2$ for every $c\neq \#$. 
It follows that $\mathcal H^{wc}_0 = \sum_{c\in\Sigma-\{\#\}} \log_2 {{n-1}\choose{n_c}} \geq \sigma - 1$. Re-arranging terms, this becomes $\sigma \leq \mathcal H^{wc}_0 + 1$. Plugging this into 
the bound $r\leq \mathcal H^{wc}_0 + \sigma$ of Theorem \ref{th:r<nHk} we obtain  our claim. \qed 
\end{proof}

Combining Lemma \ref{cor:r<nHk} and Theorem \ref{th:locate}, we obtain:

\begin{corollary}
	Let $\mathcal T$ be a trie with $n$ nodes, and let $\mathcal H^{wc}_k$ be the $k$-th order worst-case entropy of $\mathcal T$ for any $0 \leq k \leq \max\{0,\alpha\log_\sigma n-1\}$ and $0 < \alpha < 1$.
	The index of Theorem \ref{th:locate} takes $2n + o(n) + O(\mathcal H^{wc}_k\log n)$ bits of space and locates the pre-order identifiers of the $occ$ nodes reached by a path labeled with  $P\in\Sigma^m$ in $O\left(m\log\sigma + occ\right)$ time. 
\end{corollary}

% ---- Bibliography ----
%
% BibTeX users should specify bibliography style 'splncs04'.
% References will then be sorted and formatted in the correct style.
%
 \bibliographystyle{splncs04}
 \bibliography{biblio}

\end{document}